\documentclass[runningheads]{llncs}
\usepackage[T1]{fontenc}
\usepackage{graphicx}
\usepackage{wrapfig}
\usepackage{algorithm}
\usepackage{algpseudocode}
\usepackage{amsmath}
\usepackage{amssymb}
\usepackage{url,hyperref}   
\usepackage{todonotes}
\usepackage{longtable}
\usepackage[caption=false]{subfig}
\usepackage{wrapstuff}

\usepackage{thm-restate}

\newcommand{\lw}[1]{{\color{orange} Lorenz: #1}}

\newcommand{\CommentedOut}[1]{}
\newcommand{\MainProof}[1]{}

\def\orcidID#1{\href{http://orcid.org/#1}{\raisebox{-1.25pt}{\includegraphics{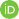}}}}

\begin{document}
\title{Positive Almost-Sure Termination of Polynomial Random Walks}
%
\author{Lorenz Winkler\orcidID{0009-0000-0177-2764} \and
Laura Kov\'acs\orcidID{0000-0002-8299-2714}}
\authorrunning{Winkler and Kov\'acs}
%
\institute{TU Wien, Vienna, Austria\\
\email{lorenz.winkler@tuwien.ac.at}
}
\date{}
\maketitle              

\begin{abstract}
The number of steps until termination of a probabilistic program is a 
random variable. Probabilistic program termination therefore requires qualitative analysis via almost-sure termination (AST), 
while also providing quantitative answers via positive almost-sure termination (PAST) on the expected number of steps until termination. While every program which is PAST is AST, the converse is not true. The symmetric random walk with constant step size is a prominent example of a program that is AST but not PAST. 

In this paper we show that a more general class of polynomial random walks is PAST. Our random walks implement a step size that is polynomially increasing in the number of loop iterations and have a constant probability $p$ of choosing either branch. 
We show that such programs are PAST when the degree of the polynomial is higher than both the degree of the drift and a threshold  $d_\text{min}(p)$. 
Our approach does not use proof rules, nor auxiliary arithmetic expressions, such as martingales or invariants. Rather, 
we establish an inductive bound for the cumulative distribution function of the loop guard, based on which PAST is proven. 
We implemented the approximation of this threshold, by combining genetic programming, algebraic reasoning  and linear programming. 
\end{abstract}
\section{Introduction}
\label{section:introduction}
Probabilistic programs extend programs written in classical programming languages by statements that draw samples from stochastic distributions, such as Normal and Bernoulli. 
The output as well as the number of steps until termination of a probabilistic program are
random variables~\cite{Kozen1985}, which makes the analysis even harder than in the nonprobabilistic setting~\cite{hark-aiming-2020}. \emph{In this paper we focus on probabilistic termination and  introduce a class of loops for which we provide a new sufficient condition for positive almost-sure termination. }

%

Probabilistic loop termination~\cite{chakarov:supermartingale-ranking-functions,fioriti-past-using-rsms-2015} requires qualitative arguments via almost-sure termination (AST), 
and quantitative answers via positive almost-sure termination (PAST) on the expected number of steps until termination. While every program which is PAST is AST, the converse is not true. The symmetric random walk with constant step size is a prominent example of a program that is AST but not PAST. 
%
%
%
Existing works for proving PAST and/or AST rely on proof rules that need auxiliary arithmetic expressions, such as invariants and martingales, over  program variables~\cite{DBLP:journals/pacmpl/MajumdarS25}. In particular, ranking super-martingales (RSMs) or lexicographical RSMs are commonly used ingredients in (P)AST analysis~\cite{agrawal2017lexicographicrankingsupermartingalesefficient}.
For probabilistic programs without nondeterminism, RSMs are a sound and complete method for proving termination~\cite{bournez:past-origin-paper,chatterjee-2016}. 

\begin{wrapfigure}{r}{0.45\textwidth}
\begin{algorithmic}[l]
\State $n \gets 0$
\State $s \gets 0$
\While{$s\geq 0$} 
    \State $n \gets n+1$
    \State $s\gets s - n + 3 \oplus_{\frac{1}{2}} s + n + 5$ 
\EndWhile
\end{algorithmic}
\caption{Program with non-trivial RSM.}\label{alg:example-non-trivial-rsm}
\end{wrapfigure}

However, finding an RSM is challenging, as shown in Figure~\ref{alg:example-non-trivial-rsm}. Here, $\gets$ denotes variable assignments and $\oplus_p$ captures probabilistic choice:   the expression on the left hand side of $\oplus_p$ is chosen with probability $p$, and the one on the right hand side  of $\oplus_p$ is chosen with probability $1-p$. 
While the program has a finite stopping time, 
to the best of our knowledge,  no RSM of this program has been found until now. As such, existing methods would fail proving PAST for Figure~\ref{alg:example-non-trivial-rsm}.  

\paragraph{Our approach.} We overcome challenges of RSM inference in PAST analysis, by identifying a class of loops, called 
\emph{polynomial random walks} (Section~\ref{section:implementation}), for which PAST
can be proven without martingale/invariant synthesis.  While polynomial random walks are less expressive than arbitrary polynomial loops, their PAST analysis can be 
automated using genetic programming, algebraic reasoning and linear programming (Section~\ref{section:implementation}).

Our approach relies on bounding the tails of random loop variables in order  to guarantee that a random variable is close to its mean.  Tail bounds~\cite{Harchol-Balter_2023} are important in providing guarantees on the probability of extreme outcomes.
Our work analyzes tail probabilities $\mathbb{P}(X\geq t)$ of the random variable $X$ summing up the steps, to establish that the probability that $X$ exceeds some value $t$ decreases fast with increasing $t$, and the deviation from $0$ is in some sense ``controlled''. 
Key to our approach is that  
variables of polynomial random walks converge to   random variables with almost Normal distributions (Lemma~\ref{lemma:almost-normal}), whose tail bounds can be approximated (Lemma~\ref{lemma:tail-bound}). 
We therefore transform polynomial random walks into programs with larger expected stopping time, where several steps are accumulated before the loop guard is checked. In other words, we connect polynomial random walk analysis to stochastic processes over  random variables with almost Normal distributions whose variance is  exponentially growing (Section~\ref{section:summation-almost-normal}). By summing up random variables of such processes, we prove that a 
%
sub-Gaussian tail bound is preserved. The 
cumulative distribution function of this summation can  tightly be approximated using an \emph{inductive bound} over random variables. %
%


By proving that an inductive bound always exists, we find that constant-probability polynomial random walks are PAST when the step size grows fast enough (Theorem~\ref{theorem:polynomial-random-walk-threshold}). Our PAST result holds  whenever (i) the degree $d$ of the polynomials is larger than the degree of the expected value of the increments, and (ii) $d$ is larger than a threshold $d_\text{min}(p)$ parametrized by the probabilistic choice probability $p$. Our work establishes that Figure~\ref{alg:example-non-trivial-rsm} satisfies this threshold (Table~\ref{table:empirical-survival-rate-exponent}), implying thus PAST of Figure~\ref{alg:example-non-trivial-rsm} without the need of an RSM.

We implemented our PAST analysis over polynomial random walks in extension of the algebraic program analysis tool \verb|Polar|~\cite{DBLP:journals/pacmpl/MoosbruggerSBK22}. To this end, we use linear programming, by relying on \verb|OR-Tools|~\cite{ortools} and the \verb|Gurobi|-solver~\cite{gurobi}, to derive inductive bounds for fixed parameters of the program transformation over polynomial random walks. We further combined \verb|Polar| with   
 a genetic algorithm, to find the best values for those parameters. Our experimental results in 
 Section~\ref{section:implementation}, give practical evidence on the tightness of our inductive bounds on polynomial random walks. Existence of these inductive bounds imply thus PAST of the programs that are analyzed.

\paragraph{Our contributions.} We translate the problem of verifying PAST into the problem of tightly approximating tail bounds of random loop variables. We bring the following contributions\footnote{with  detailed proofs of our results given in the Appendix}:  

\begin{itemize}
\item We introduce the class of polynomial random walks for which we provide sufficient conditions to determine PAST. These conditions do not require user-provided invariants and/or martingales. We determine PAST above a threshold $d_\text{min}(p)$, which in turn only depends
on the polynomial degree $d$ of the loop updates and the probabilistic choice $p$ in the loop (Section~\ref{section:polynomial-random-walks}). 
\item We show that such a threshold $d_\text{min}(p)$ always exists by transforming polynomial random walks into stochastic processes over almost Normal variables (Section~\ref{section:summation-almost-normal}). We prove that such processes admit an inductive bound over their cumulative distribution function, allowing us to tightly approximate our  threshold $d_\text{min}(p)$.  
\item We implemented our approach to approximate $d_\text{min}(p)$, and hence conclude PAST, in extension of 
the \verb|Polar| framework (Section~\ref{section:implementation}). Our experiments showcase the tightness of our approximation, implying thus PAST. 
\end{itemize}

\section{Almost Normal Variables and Conditioning}
\label{section:summation-almost-normal}

\begin{wrapstuff}[r, width=0.45\textwidth]
{
\setlength{\intextsep}{0pt}
\setlength{\belowcaptionskip}{10pt}
\begin{figure}[H]
\begin{algorithmic}[l]
\State $n \gets 0$
\State $z \sim \mathcal{N}_{c_0}^{\delta_1,C_1}(0, \sigma_0)$
\State $s \gets s_0+ z$
\While{$s\geq F_{S_n}^{-1}(\epsilon)$} 
    \State $n \gets n+1$
    \State $z \sim \mathcal{N}_{c_0}^{\delta_1,C_1}(0, \sigma_n)$
    \State $s\gets s + z$
\EndWhile
\end{algorithmic}
\caption{Probabilistic programs  summing up almost normally distributed variables}\label{alg:cap}
\end{figure}
}
\end{wrapstuff}

This section introduces a special class of stochastic processes (Section~\ref{subsection:stoch_process}) and establishes  bound properties for such processes. The  stochastic processes we consider
are given in  Figure~\ref{alg:cap}, for which we  show that they induce random variables with sub-Gaussian tail bounds (Section~\ref{subsection:tail-bound}) whose inductive bounds can tightly be approximated (Section~\ref{subsection:inductive-bound}). 
In Section~\ref{section:polynomial-random-walks}
we then prove that polynomial random walks can be transformed into the stochastic process of Figure~\ref{alg:cap}, allowing to conclude PAST of polynomial random walks (Theorem~\ref{theorem:polynomial-random-walk-threshold}). 

In the sequel, we respectively denote by $\mathbb{N}, \mathbb{R}, \mathbb{R}^+$ the set of natural, real, and positive real (including zero) numbers. We reserve $n\in\mathbb{N}$ for the loop iteration (counter).  We assume familiarity with probabilistic programs and their semantics, and refer to~\cite{Durrett2019,Harchol-Balter_2023} for details. The probability measure is denoted by $\mathbb{P}$, while we use $\mathbb{E}$  to denote the expected values of random variables. Further, we denote with $(X | E)$ the conditional probability distribution of $X$ given event $E$.


\subsection{Almost Normally Distributed Loop Summations}\label{subsection:stoch_process}
We consider stochastic processes induced by the  probabilistic program of Figure~\ref{alg:cap}, which  uses the random variable series $\{S_n\}_{n\in\mathbb{N}}$ and $\{Z_n\}_{n\in\mathbb{N}}$ corresponding to the values of  $s$ and $z$ in the $n$-th iteration.  We respectively denote by $F_{S_n}$ and $F_{Z_n}$ the cumulative distribution function (cdf) of $S_n$ and $Z_n$. 
The notation $z\sim \mathcal{N}_{c_0}^{\delta_1,C_1}(0,\sigma_n)$ indicates that $z$ is drawn from a distribution that is almost equivalent to a Normal distribution $\mathcal{N}(0,\sigma_n)$ with variance $\sigma_n^2$, in a sense, that $c_0$ bounds the absolute deviation of cdf of $Z_n$ from the cdf of the Normal variable. Additionally, $\delta_1$ bounds the shift of the sub-Gaussian tail bound and $C_1$ is the multiplicative deviation of that bound's variance (see Lemma~\ref{definition:z_n}). 
The initial value of $s$ is also drawn from such an almost Normal distribution with variance $\sigma_0^2$. In each loop iteration $n$, a sample is drawn from an almost Normal distribution with variance $\sigma_n^2$ and then added to $s$. The loop is exited, once $s$ is within the smallest $\epsilon$-fraction of $S_n$.

Recall that the  stochastic processes induced by $\{S_n\}$ and $\{Z_n\}$ represent a Markov chain~\cite{Durrett2019}. 
Based on the semantics of Figure~\ref{alg:cap}, the series $S_n$ sums variables $Z_0,\dots,Z_n$, as follows:  
%
 
\begin{itemize}
    \item $S_0 = Z_0$;
    \item  $S_{n+1} = (S_n \mid S_n\geq F_{S_n}^{-1}(\epsilon)) + Z_{n+1}$, where $F_{S_n}^{-1}$ is the inverse of the cdf of $S_n$. 
    \end{itemize}

Note that for $S_n$ only the paths are considered, for which the program has not yet terminated. To reason about termination of Figure~\ref{alg:cap}, we rely on a right tail bound of the random variable $S_n$, as it gives an upper limit on the probability of a random variable exceeding a certain value~\cite{Harchol-Balter_2023}. In other words, the right tail bound of $S_n$ quantifies how likely it is to observe values of $S_n$ in the extreme right (or upper) tail of a distribution. While such extreme cases might yield to the non-termination of Figure~\ref{alg:cap}, in Section~\ref{subsection:tail-bound} we show that a tail bound for $S_n$ can be derived from the tail bound of $Z_n$, which ensures that this behavior is unlikely. 
Moreover, by adjusting  standard properties of distributions~\cite{Durrett2019}, we also bound the behavior of $Z_n$, as listed below.

\begin{lemma}[CDF deviation and bound]
    \label{definition:z_n}
        Let $\{Z_n\}_{n\in \mathbb{N}}$ be a sequence of random variables, as defined in Figure~\ref{alg:cap}. Recall that  $Z_n$  follows an almost Normal distribution $\mathcal{N}_{c_0}^{\delta_1,C_1}(0,\sigma_n)$. The following holds:  
        \begin{enumerate}
            \item The cdf of $Z_n$, denoted as $F_{Z_n}$ deviates from a Normal distribution only by at most $c_0$. That is,  $|F_{Z_n}(z) - \Phi(\frac{z}{\sigma_n})| \leq c_0$, where $\Phi$ denotes the  cumulative distribution function of the Normal distribution.
            \item The variable $Z_n$ admits a sub-Gaussian tail bound on its right tail, which is offset by $\delta_1\sigma_n$ and has  variance 
            of $\frac{\sigma_n^2}{C_1}$. That is, \[\forall a\in \mathbb{R}^+:\mathbb{P}(Z_n \geq a+\delta_1\sigma_n)\leq \exp\left(C_1\frac{-a^2}{2\sigma_n^2}\right).\]
        \end{enumerate}
    \end{lemma}

    \subsection{Tail Bound for $S_n$}
    \label{subsection:tail-bound}
    Lemma~\ref{definition:z_n} limits the behavior of variable $z$ in Figure~\ref{alg:cap}. We next show that, in addition to $z$, also variable $s$  does not grow in an uncontrolled way.  Similarly to  Lemma~\ref{definition:z_n}, we reason about the right tail of $S_n$ and prove that it admits a sub-Gaussian tail bound (Lemma~\ref{lemma_upper_bound_s_n_right_tail}); as such, the probability of $S_n$ being significantly greater than its expected value is limited.

\begin{restatable}[Preservation of sub-Gaussian tail bound]{lemma}{lemmaUpperBoundSnRightTail}
\label{lemma_upper_bound_s_n_right_tail}
Consider the random variable series $\{S_n\}_{n\in\mathbb{N}}$ induced by $s$ in  Figure~\ref{alg:cap}. Assume that the variance 
$\sigma_n^2$ grows  such that the inequality $\sigma_{n+1}^2 \geq d (\sigma_1^2 + \dots + \sigma_n^2)$ holds for some $d\in\mathbb{R}^+\setminus\{0\}$. Then,  $S_n$ admits a sub-Gaussian (but not centered) tail bound: 
        \[
        \forall a\in \mathbb{R}^+:
        \mathbb{P}\bigg(S_n \geq a + (\frac{\sqrt{1+d}}{\sqrt{1+d}-1}\delta_1+ b) \sqrt{\sum_{i=0}^n\sigma_i^2}~\bigg) ~\leq ~\exp\left(C_1\frac{-a^2}{2\sum_{i=0}^n\sigma_i}\right),\]
        where
        $b \geq\frac{ \sqrt{2 \ln{\frac{1}{1-\epsilon}}}}{\sqrt{C_1}(\sqrt{(1+d)}-1)}$.
\end{restatable}

\subsection{Inductive Bound Set for $S_n$}
\label{subsection:inductive-bound}
While Lemma~\ref{lemma_upper_bound_s_n_right_tail} gives an upper bound for the cdf  of $S_n$, the sub-Gaussian tail bound of $S_n$ is not very sharp for values close to the mean of $S_n$. In this section we extend Lemma~\ref{lemma_upper_bound_s_n_right_tail} with a \emph{union-bound compositional approach} to tighten cdf bounds, as follows.  We (i) split the cdf into $m\in\mathbb{N}$ pieces, (ii) provide lower bounds $B(S_n)$ for each piece of the cdf, and (iii) combine the lower bounds inductively into a tighter upper bound for the cdf of $S_n$. 

Our compositional framework is inductive over the lower bounds of cdf pieces: $S_0$  satisfies the bound $B(S_0)$ (base case) and, if $S_{n}$ satisfies the bound $B(S_n)$, then  $S_{n+1}$ satisfies $B(S_{n+1})$ (induction step). 
Our bound $B(S_n)$ is uniquely defined by a set of inequalities using two vectors $\vec{a}_B,\vec{b}_B$ of bounding values, where 
elements of $\vec{a}_B,\vec{b}_B$ provide the location of and lower bounds on the cdf pieces of $S_n$. As such, we set: 
{\small\begin{equation} \label{eq:inductiveBoundSet}
{ 
B(S_n)=\bigg\{\mathbb{P}\big(S_n\leq \vec{a}_{B,1}\sqrt{\sum_{i=0}^{n}\sigma_i^2}\big)\geq \vec{b}_{B,1}, 
~~\dots,~~\allowbreak 
\mathbb{P}\big(S_n\leq \vec{a}_{B,m}\sqrt{\sum_{i=0}^{n}\sigma_i^2}\big)\geq \vec{b}_{B,m}\bigg\}.\hspace*{-1em}
}
\end{equation}}
For simplicity, we assume $\vec{a}_1\leq0$ and $\vec{b}_1\geq \epsilon$, in order to ensure that only negative values of $s$ exit the loop therefore enforcing $F^{-1}_{S_n}(\epsilon)\leq 0$.
If each inequality in $B(S_n)$ is valid, we say that the  bound $B(S_n)$ \emph{holds}.
By simple arithmetic reasoning, we state the following property over bounds. 
\begin{lemma}[Partial order of bounds]\label{lem:bound:order}
    The bounds $B(S_n)$ admit an ordering whenever they describe the same intervals and  probabilities are ordered. That is: 
    \[B' (\cdot)\leq B(\cdot) \iff \vec{a}_{B'} = \vec{a}_{B} \land \forall_{1\leq i \leq m}: b_{B',i}\geq b_{B,i}\]
\end{lemma}

\paragraph{\bf Inductive computations of bound set for $S_n$.} With regard to the partial order, we provide our inductive computation for bounds of $S_n$: from a bound $B(\cdot)$ that holds for $S_n$, we compute a bound $B'(\cdot)$ that holds for $S_{n+1}$. If $B'(\cdot)\leq B(\cdot)$,  then $B(S_m)$ holds for all $m\geq n$, as  our computation of new bounds ensures monotonicity of bounds. Doing so and starting with $B(S_n)$,  (i)  the left tail with a cdf of $\epsilon$ is cut away from $S_n$, yielding $S'_n$. Then, using Lemma~\ref{definition:z_n}, we (ii) add an  almost Normal variable  with variance $\sigma_{n+1}^2 \geq d\sum_{i=0}^{n}\sigma_i^2$  to the resulting distribution of $S'_n$, and compute a new bound $B'(S_{n+1})$. 
Our bound computation uses a union-bound approach for deriving  interval boundaries. In addition to the bounds from $B$, we use  tail bounds from Lemma~\ref{lemma_upper_bound_s_n_right_tail}, as otherwise obtaining an inductive bound for $S_n$ with $\epsilon>0$ is not possible.

The next example illustrates our  inductive bound set computation for $S_n$.

\begin{example}
    Consider an instance of the stochastic process
    of Figure~\ref{alg:cap}, by setting $\epsilon = 0.1$ and $d=3$ and using a set of  almost Normal variables $\{Z_n\}_{n\in \mathbb{N}}$ with parameters $C_1=1$, $\delta_1=0$, and $c_0=10^{-3}$. For this instance of Figure~\ref{alg:cap}, we define the bound set: \[B(S_n)=\bigg\{\mathbb{P}\big(S_n \leq 0\big ) ~\geq 0.1, ~~~\mathbb{P}\big(S_n \leq \sqrt{\sum\nolimits_{i=0}^{n}\sigma_i^2}~\big) ~\geq 0.4\bigg\}.\]
 \noindent   For the tail bounds of Lemma~\ref{lemma_upper_bound_s_n_right_tail},  we (arbitrarily) pick the value $3\sqrt{\sum\nolimits_{i=0}^{n}\sigma_i^2}$ and derive: 
    \[\mathbb{P}\big(S_n\geq 3\sqrt{\sum\nolimits_{i=0}^{n}\sigma_i^2}~\big) ~~\leq 
    \exp(\frac{-(2.54)^2}{2})~~\leq 0.04.\] Therefore, $\mathbb{P}(S_n\leq 3\sqrt{\sum\nolimits_{i=0}^{n}\sigma_i^2})\geq 0.96$.
    Using these inequalities, we  compute new bounds $S_n'$. Here,  $S_n'=(S_n\geq F_{S_n}^{-1} (\epsilon))$ is the variable obtained from $S_n$ by cutting away  the left tail with weight $\epsilon$. 
    For readability, we denote the summation of variances up to $S_n$ through $\sigma_{S_n}^2:=  \sum_{i=0}^{n}\sigma_i^2$, and similarly $\sigma_{S_{n+1}}^2:=  \sum_{i=0}^{n+1}\sigma_i^2$. With this notation at hand, we have: 
    \begin{align*}
        \mathbb{P}\big(S_{n+1} \leq 0\big) ~~\geq ~~&
        \mathbb{P}\big(S'_n\leq \sigma_{S_{n}}\big)\cdot\mathbb{P}\big(Z_{n+1}\leq-\sigma_{S_{n}}\big)+\nonumber
        \\ &\mathbb{P}\big(\sigma_{S_{n}}\leq S'_n\leq 3\sigma_{S_{n}}\big)\cdot\mathbb{P}\big(Z_{n+1}\leq -3\sigma_{S_{n}}\big)\\
        \mathbb{P}\big(S_{n+1} \leq \sigma_{S_{n+1}}\big) ~~\geq ~~&
        \mathbb{P}\big(S'_n\leq \sigma_{S_{n}}\big)\cdot\mathbb{P}\big(Z_{n+1}\leq \sigma_{S_{n+1}}-\sigma_{S_{n}}\big)+\nonumber
        \\ &\mathbb{P}\big(\sigma_{S_{n}}\leq S'_n\leq 3\sigma_{S_{n}}\big)\cdot\mathbb{P}\big(Z_{n+1}\leq\sigma_{S_{n+1}} -3\sigma_{S_{n}}\big)
    \end{align*}

 \noindent   Note that $ \sqrt{\sum_{i=0}^{n+1}\sigma_i^2} \geq \sqrt{1+d}  \sqrt{\sum_{i=0}^{n}\sigma_i^2}$ and $Z_{n+1}$ has standard deviation $\sqrt{3\sum_{i=0}^{n}\sigma_i^2}$. Using the bound set $B(S_n)$ and Lemma~\ref{definition:z_n}, we get:
    \begin{eqnarray*}
        \mathbb{P}\bigg(S_{n+1} \leq 0\bigg) ~\geq~\frac{0.4-\epsilon}{1-\epsilon}\bigg(\Phi\big(\frac{-1}{\sqrt{3}}\big)-c_0\bigg)+
        \frac{0.56}{1-\epsilon}\bigg(\Phi\big(\frac{-3}{\sqrt{3}}\big)-c_0\bigg)
        \approx 0.1188
\end{eqnarray*}
    
  \noindent  Similarly,  for the second bound $\mathbb{P}(S_{n+1}\leq\sqrt{\sum_{i=0}^{n+1}\sigma_i^2}) \gtrapprox  0.4138$.
    The bound $B$ is inductive as the new bound which we computed for $S_{n+1}$ is smaller than the initial bound $B$ which is assumed to hold for $S_n$.
    \hfill \ensuremath{\triangle}
\end{example}

\paragraph{\bf On the existence of inductive bounds.} Our approach to 
inductively computing bound sets for $S_n$ relies on an union-bound argument to improve the cdf bound of Lemma~\ref{lemma_upper_bound_s_n_right_tail}.  Recall that the sub-Gaussian tail bound of $S_n$ only depends on $d\in\mathbb{R^+}\setminus{0}$ and $\epsilon\in\mathbb{R}^+\setminus 0$.  We next show that the existence of an inductive bound set is conditioned only by such a $d$. Namely, Theorem~\ref{theorem:epsilon-exists-always} ensure that, for every $d$ there is an $\epsilon\in\mathbb{R^+}\setminus{0}$, such that an inductive bound exists for the corresponding series $\{S_n\}_{n\in\mathbb{N}}$, with   $F_{S_n}^{-1}(\epsilon)\leq0$. The probability of $S_n$ being smaller than zero is therefore \emph{lower bounded} by some nonzero percentage.

\begin{restatable}[Inductive bound set]{theorem}{theoremEpsilonExistsAlways}
\label{theorem:epsilon-exists-always}
    For every $d\in\mathbb{R^+}\setminus{0}$ there exists an $\epsilon\in\mathbb{R^+}\setminus{0}$ such that an inductive bound set $B(S_n)$ holds, with $F_{S_n}^{-1}(\epsilon)\leq0$, given that $c_0$ converges to $0$ and can be chosen arbitrarily small.
\end{restatable}

\section{Polynomial Random Walks}
\label{section:polynomial-random-walks}

\begin{wrapstuff}[r, width=0.45\textwidth]
{
\setlength{\intextsep}{0pt}
\setlength{\belowcaptionskip}{10pt}
\begin{figure}[H]
    \begin{algorithmic}[l]
    \State $n\gets 0$
    \State $y \gets y_0$
    \While{$y > 0$}
        \State $n\gets n + 1$
        \State $x\gets q_1[n] \oplus_{p} q_2[n]$
        \State $y\gets y + x$
    \EndWhile
    \end{algorithmic}
        \caption{Polynomial random walk $\mathcal{P}$}
    \label{alg:random_walk}
\end{figure}
}
\end{wrapstuff}

Section~\ref{section:summation-almost-normal} showed that stochastic processes defined by Figure~\ref{alg:cap} have bounded behavior, allowing us to lower bound the termination probability via sub-Gaussian tail bounds and inductive bound sets. In this section we map the termination analysis of certain polynomial programs, called polynomial random walks,  to the framework of Section~\ref{section:summation-almost-normal}. Importantly,  \emph{we reduce the problem of verifying PAST of polynomial random walks to the problem of ensuring existence of inductive bounds} (Theorem~\ref{theorem:polynomial-random-walk-threshold}). Our recipe  consists of  transforming a polynomial random walk program $\mathcal{P}$ to a program that (i) bounds PAST of $P$ and (ii) is equivalent to the stochastic process of Figure~\ref{alg:cap}. 

\subsection{Programming Model}\label{sec:pgmModel} 
We define the class of  \emph{polynomial random walks} via the programming model of Figure~\ref{alg:random_walk}, where $q_1[n],q_2[n]\in \mathbb{R}[n]$ are arbitrary polynomial expressions in the loop counter $n$.  The \emph{degree of a polynomial random walk program} $\mathcal{P}$, written as $\deg(\mathcal{P})$, is  given by the maximum degree of its  polynomials, that is  $\deg(\mathcal{P}) = \max\{\deg (q_1[n]), \deg (q_2[n])\}$.
The  series $\{X_n\}_{n\in\mathbb{N}}$, $\{Y_n\}_{n\in\mathbb{N}}$ induced by the random loop variables $x,y$ are next defined. 

\begin{definition}[Random walk variables]
\label{definition:x_n}\label{definition:y_n}
    The \emph{random walk variable $X_n$} corresponding to the loop variable $x$ at iteration $n$ in Figure~\ref{alg:random_walk} is 
    \begin{equation*}
        X_n = \begin{cases}
            q_1[n] & \text{with probability } p\\
            q_2[n] & \text{with probability } 1-p.\\
        \end{cases}
    \end{equation*}
  The \emph{random walk variable $Y_n$} captures the distribution of $y$ after iteration $n$, as:
    \[Y_{n+1} = (Y_n | Y_n > 0) + X_{n+1}.\]
\end{definition}
The second-order moment of a random variable $X_n$ is written as $Var(X_n)$. For Figure~\ref{alg:random_walk}, we have $\mathbb{E}(X_n) = q_1[n] p + q_2[n] (1-p)$ and $Var(X_n) = q_1[n]^2 p + q_2[n]^2 (1-p)$, capturing the   mean (first moment) and variance (second moment) of $X_n$; note that both moments of $X_n$ are also polynomials in $n$.


To prove PAST of Figure~\ref{alg:random_walk} we need to prove that the expected value of its stopping time is finite~\cite{hark-aiming-2020}. Based on the semantics of Figure~\ref{alg:random_walk}, it is easy to see that the stopping time of Figure~\ref{alg:random_walk} is given by the first iteration $n$ in which $Y_n$ becomes negative. 
\begin{definition}[Expected stopping time]
\label{definition:stopping-time}
    Let $T$ be $\inf\{n\geq0:Y_n\leq0\}$, where $T$ denotes  the \emph{stopping time} of the stochastic process induced by the polynomial random walk of Figure~\ref{alg:random_walk}. The \emph{expected stopping time} of Figure~\ref{alg:random_walk} is defined as $\mathbb{E}(T) = \sum_{n=0}^{\infty} \mathbb{P}(T\geq n)$.
\end{definition}

\begin{wrapstuff}[r, width=0.45\textwidth]
{
\setlength{\intextsep}{0pt}
    \begin{figure}[H]
    \begin{algorithmic}
    \State $n\gets 0$
    \State $y \gets y_0$
    \While{$n\leq n_0$}
        \State $n\gets n+1$
        \State$x\gets q_1[n] \oplus_{p} q_2[n]$
        \State $y\gets y + x$
    \EndWhile
    \While{$y > g$} \Comment{\textit{where $g \leq 0$}}
        \State $z\gets 0$
        \State $n' \gets n$
        \While{$n \leq n'\cdot k$}
            \State $n\gets n + 1$
            \State$x\gets q_1[n] \oplus_{p} q_2[n]$
            \State $z\gets z + x$
        \EndWhile
        \State $y\gets y + z$
    \EndWhile
    \end{algorithmic}
        \caption{Transformed random walk}
    \label{alg:random_walk_transformed}
    \end{figure}
}
\end{wrapstuff}
We  exploit Definition~\ref{definition:stopping-time} to show that Figure~\ref{alg:random_walk} is PAST under additional conditions. Namely, we translate Figure~\ref{alg:random_walk} into  Figure~\ref{alg:random_walk_transformed} and ensure that the stopping time of Figure~\ref{alg:random_walk_transformed}  becomes finite above a certain threshold; this threshold   depends only on the maximum polynomial degree of Figure~\ref{alg:random_walk} and the variable $k$. We then show that finiteness of the stopping time of Figure~\ref{alg:random_walk_transformed} implies PAST of  Figure~\ref{alg:random_walk} (Lemma~\ref{lemma:stopping-time-inequality}). 

\paragraph{Program transformation.} We  translate Figure~\ref{alg:random_walk} into the stochastic process of Figure~\ref{alg:random_walk_transformed}. 
This program transformation is defined through the parameters $n_0, k$ and $g$. The loop body of Figure~\ref{alg:random_walk_transformed} is initially executed several times, accumulating $n_0$ steps. In every iteration of the outer loop $k$ times as many steps as before are summed up, before the loop guard is checked again. Furthermore, the loop guard of Figure~\ref{alg:random_walk_transformed} might be relaxed, as $g\leq 0$. We highlight similarities between 
Figure~\ref{alg:random_walk_transformed} and the summation of almost Normal variables with conditioning in Figure~\ref{alg:cap}: the inner loop of Figure~\ref{alg:random_walk_transformed} computes the value for $z$ by summing up $X_i$. As argued in Section~\ref{subsection:summation-polynomial-random-walk-increments}, this is similar to drawing $z$ from an almost Normal distribution as in Figure~\ref{alg:cap}.

    We have that the expected stopping time of Figure~\ref{alg:random_walk_transformed} is  larger than of Figure~\ref{alg:random_walk}.

\begin{restatable}[Stopping Time Inequality]{lemma}{lemmaStoppingTimeInequality}
    \label{lemma:stopping-time-inequality}
 Let $T'$ be $\inf\{n\geq n_0:Y_n\leq g\}$, denoting the stopping time of Figure~\ref{alg:random_walk_transformed}. Then, $\mathbb{E}(T) \leq \mathbb{E}(T')$.
\end{restatable}
We denote with $p_\text{term}$ a lower bound for the probability of the outer loop terminating. 
In what follows, we will ensure the stopping time of the program in Figure~\ref{alg:random_walk_transformed} is finite when the probability $p_\text{term}$ is high enough and $k$ is small 
(Lemma~\ref{lemma:bound-from-termination-probability}).  The existence of a nonzero lower bound for $p_\text{term}$ is implied by Theorem~\ref{theorem:epsilon-exists-always}; we note that  $p_\text{term}$ depends on the probability of choosing a branch $p$ and grows as $k^{2\deg{(\mathcal{P}})}+1$ increases. By setting  $k$ to its maximum value, 
we derive a threshold $d_\text{min}(p)$ for the degree $\deg(\mathcal{P})$  of the polynomial random walk of Figure~\ref{alg:random_walk_transformed} and prove that the stopping time of Figure~\ref{alg:random_walk} above this threshold $d_\text{min}(p)$  is finite (Theorem~\ref{theorem:polynomial-random-walk-threshold}). We thus use  $d_\text{min}(p)$ to provide sufficient conditions for deciding PAST of the polynomial random walks in Figure~\ref{alg:random_walk}.

\subsection{Loop Summations of Polynomial Random Walk Increments}
\label{subsection:summation-polynomial-random-walk-increments}
We now establish the formal connection between the polynomial random walks of Figure~\ref{alg:random_walk_transformed} and the stochastic processes of Figure~\ref{alg:cap}. We prove that the loop summation (defined below) of the increments of the random walk in Figure~\ref{alg:random_walk_transformed} is almost normally distributed as given in  Lemma~\ref{definition:z_n}, when an inequality over the degrees of expected value of the step and its variance is true. This inequality holds, whenever the leading terms of the steps cancel out.

\begin{definition}[Random walk loop summation]\label{def:random:walk:sum}
    The random variables $U_{0} = y_0+X_{0} + \dots + X_{n_0}$ and $U_{n'} =  X_{n'} + \dots + X_{\lceil n'\cdot k\rceil}$ are  \emph{(loop) summations} of the random variables $X_i$ of Figure~\ref{alg:random_walk_transformed}.
\end{definition}

\noindent Lemma~\ref{lemma:almost-normal}  then shows that the absolute deviation $c_0$ from the cdf of the Normal distribution  converges to $0$. Further,  Lemma~\ref{lemma:tail-bound} conjectures that the summation of random walk increments admits a sub-Gaussian tail bound with $C_1=4p(1-p)$ and $\delta_1$ converging to $0$, thus establishing, that the loop summation follows an almost Normal distribution $\mathcal{N}_{c_0}^{\delta_1,4p(1-p)}$.

\begin{restatable}[Convergence of cdf deviation]{lemma}{lemmaSummationCdfDeviation}
    \label{lemma:almost-normal}
   Assume that $\deg(Var(X_i)) > 2\deg(\mathbb{E}(X_i)) + 1$ holds for Figure~\ref{alg:random_walk_transformed}.  Then,  the normalizations of the loop summations $U_0$ and $U_{n'}$ follow a Normal distribution up to a constant error $c_0$, with $c_0$  converging to $0$ with increasing $n_0$:
    \begin{equation}\label{eq:cdf:convergence}\forall n'\geq n_0: 
    \left|F_{U_{n'}} (z) - \Phi\left( \frac{z}{\sqrt{\sum_{i=n_0}^{\lceil n'\cdot k \rceil} Var(X_i)}}\right)\right| \leq c_0
    \end{equation}
    where $n_0$ is as given in Figure~\ref{alg:random_walk_transformed} and $F_{U_n'}$ denotes the cdf of $U_{n'}$.
\end{restatable}

Using Lemma~\ref{lemma:almost-normal}, we derive that the loop summations of polynomial random walks follow an almost Normal distribution, similarly to the stochastic process of Figure~\ref{alg:cap}  in Section~\ref{section:summation-almost-normal}. 

\begin{restatable}[Tail bound for $U_{n'}$]{lemma}{lemmaSummationTailBound}
    \label{lemma:tail-bound}
    Let $\sigma_{U_{n'}}$ be the standard deviation of $U_{n'}$ and assume  $\deg(Var(X_i)) > 2\deg(\mathbb{E}(X_i)) + 1$. Then,  the right tail probability is bounded with $\delta_1$ converging to $0$, as follows:
  \[  \begin{array}{lcll}
        \textstyle\mathbb{P}(U_{n'} \geq \lambda \sigma_{U_{n'}} + \delta_1\sigma_{U_{n'}}) &\leq&\exp\left(4(1-p)p\frac{-\lambda^2}{2}\right), & \text{and}\\\textstyle
                \mathbb{P}(U_{0} \geq \lambda \sigma_{U_{0}} +\delta_1\sigma_{U_{0}}) &\leq& \exp\left(4(1-p)p\frac{- \lambda^2}{2}\right).&
    \end{array}\]
\end{restatable}

%
\begin{example}\label{ex:RSM:bounded}
   Consider our motivating example from Figure~\ref{alg:example-non-trivial-rsm}. 
   In order to ensure that its loop summations follow an almost Normal distribution, with $c_0$ and $\delta_1$ converging to zero, we need to ensure that $\deg(Var(X_i)) > 2\deg(\mathbb{E}(X_i)) + 1$. This inequality is true, since $Var(X_i) = (i+1)^2$ and $\mathbb{E}(X_i) = 4$, hence $\deg(Var(X_i)) = 2$ and $\deg(\mathbb{E}(X_i))=0$. Consequently, Lemma~\ref{lemma:almost-normal}~and~\ref{lemma:tail-bound} can be used.
    \qed
\end{example}

In the remaining, we define the random variable series $\{Z_n\}_{n\in\mathbb{N}}$   corresponding to the loop summation of the inner loop of Figure~\ref{alg:random_walk_transformed}. That is, 
$Z_n$ captures the program variable $z$ at the end of every iteration of the outer loop of Figure~\ref{alg:random_walk_transformed}, with $Z_0$ being the variable corresponding to $z$ after its first loop. As such, 
\begin{itemize}
    \item $Z_0=U_0$, and 
    \item $Z_n=U_{n'_{(n)}}$, where $n'_{(n)}$ is the value of $n'$ in the $n$-th iteration of the outer loop of Figure~\ref{alg:random_walk_transformed}. 
\end{itemize}
Further, $Y_n$ is induced by the program variable $y$ of Figure~\ref{alg:random_walk_transformed}, capturing the loop summation of $Z_n$ with repeated conditioning.
In order to use inductive bound sets as in Theorem~\ref{theorem:epsilon-exists-always}, the variance of $\{Z_n\}_{n\in\mathbb{N}}$ must grow consistently and exponentially. This is however clearly ensured by choosing $k>1$ in Figure~\ref{alg:random_walk_transformed}, implying the following result\footnote{see Appendix~\ref{appendix:proof-variance-growth} for proof}. 

\begin{restatable}[Growth of variance]{lemma}{lemmaVarianceGrowth}
    \label{lemma:variance-growth-polynomial-random-walk}
    The variance $\{\sigma^2_n\}_{n\in\mathbb{N}}$ of $\{Z_n\}_{n\in\mathbb{N}}$ grows exponentially, with $\delta'$ converging to $1$:
    \[\sigma_{n+1}^2 \geq \left(\delta'k^{2\deg(\mathcal{P})+1}-1)\right) \sum_{i=0}^n \sigma_i^2\]
\end{restatable}

Lemmas~\ref{lemma:almost-normal}~and~\ref{lemma:tail-bound}  establish that $Z_n$ follows an almost Normal distribution as in Lemma~\ref{definition:z_n}. Together with 
Lemma~\ref{lemma:variance-growth-polynomial-random-walk}, this ensures that the right tail of $Y_n$ can be bounded (Lemma~\ref{lemma_upper_bound_s_n_right_tail}), and therefore inductive bounds can be used. Based on this bounds, Section~\ref{subsection:survival_function} introduces conditions on the stopping time $T$ of Figure~\ref{alg:random_walk_transformed} being finite, implying thus PAST of 
Figure~\ref{alg:random_walk}. 

\subsection{Bounding the Stopping Time and PAST}\label{subsection:survival_function}
Recall that, using Definition~\ref{definition:stopping-time},  the expected stopping time $\mathbb{E}(T)$ of 
Figure~\ref{alg:random_walk_transformed} is determined by the loop summation variables $Y_n$ and is set to:
\[
\mathbb{E}(T)=\sum_{n=0}^\infty\mathbb{P}(T\geq n).\]
Using Lemma~\ref{lemma:variance-growth-polynomial-random-walk}, we obtain the following bound on $\mathbb{P}(T\geq n)$, and hence on $\mathbb{E}(T)$.

\begin{restatable}[Bounding the stopping time]{lemma}{lemmaBoundFromTerminationProb}
\label{lemma:bound-from-termination-probability}
    Assume that the outer loop of Figure~\ref{alg:random_walk_transformed} terminates with probability $p_\textnormal{term}$ after some $n_0$. Then, \\
    \[\mathbb{P}(T\geq n) \leq \min\left\{1,B n^{\frac{\ln(1-p_\textnormal{term})}{\ln(k+\frac{1}{n_0})}} \right\} \]
    where $B=\frac{1}{(1-p_\textnormal{term})^{\log_{k+\frac{1}{n_0}}(n_0)+2}}$.
    Therefore, if
     $\ln(1-p_\textnormal{term})<-\ln(k+\frac{1}{n_0})$ holds, then the expected stopping time $\mathbb{E}(T)$ is finite. 
\end{restatable}

\paragraph{\bf On the finiteness of stopping times.} 
Lemma~\ref{lemma:bound-from-termination-probability} formulates  conditions under which Figure~\ref{alg:random_walk_transformed}
has finite stopping time. These conditions effectively only depend on the probability $p_\text{term}$ and $k$, as $n_0$ can be chosen arbitrarily. As such, \emph{finiteness of $\mathbb{E}(T)$ and PAST of Figure~\ref{alg:random_walk_transformed} is reduced to finding an inductive bound}, with $d=\delta'k^{2\deg(\mathcal{P})+1}-1$, $C=p(1-p)$ and $\epsilon$ (which is a lower bound for $p_\text{term}$) so large, that the inequality in Lemma~\ref{lemma:bound-from-termination-probability} is satisfies. The terms $\delta_1, \delta'$ and $c_0$ can be computed from a finite, arbitrary $n_0$.

To this end, 
let  $p_\text{i.b.}(d, n_0, p)$  be a  to-be-determined function that returns the largest $\epsilon$ for which an inductive bound exists, which is a lower bound for $p_\text{term}$. Then, 
\begin{equation}
\label{equation:optimization-problem-random-walk}
    \mathbb{P}(T\geq n) \leq \inf_{1<k}\left\{ \inf_{0\leq n_0}\left\{\min\left\{Bn^{\frac{\ln(1-p_\text{i.b.}(\delta'k^{2\deg\{\mathcal{P}+1\}}-1, n_0, p))}{\ln(k+\frac{1}{n_0})}}  \right\}\right\}\right\}
\end{equation}
with $B=(1-p_\text{i.b.}(\delta'k^{2\deg\{\mathcal{P}+1\}}-1, n_0, p))^{-(\log_{k+\frac{1}{n_0}}(n_0)+2)}$.
Enforcing~\eqref{equation:optimization-problem-random-walk} requires however \emph{solving a non-trivial optimization problem}: we need to  approximate the function $p_\text{i.b.}(d,n_0,p)$. While in Section~\ref{section:implementation} we show that this approximation can be done using linear programming and a genetic algorithm, the statement of~\eqref{equation:optimization-problem-random-walk} has theoretical consequences. The existence of an inductive bound set from Theorem~\ref{theorem:epsilon-exists-always} implies that 
an $\epsilon$ for $p_\text{i.b.}(d,n_0, p)$ always exist, allowing us to state a PAST condition over  polynomial random walks $\mathcal{P}$ from Figure~\ref{alg:random_walk}. 

\begin{restatable}[PAST of polynomial random walks]{theorem}{theoremPolynomialRandomWalkThreshold}
    \label{theorem:polynomial-random-walk-threshold}
    Let $\mathcal{P}$ be a polynomial random walk program of Figure~\ref{alg:random_walk}. 
    For every  probabilistic choice $p$ in $\mathcal{P}$ there exists a threshold $d_\textnormal{min}(p)$ such that $\mathcal{P}$ has finite expected stopping time, when $\deg(\mathcal{P}) >d_\textnormal{min}(p)$ and $\deg(\mathcal{P})>\deg(p(q_1[n])+(1-p)q_2[n])$.
\end{restatable}

\begin{example}[PAST of Figure~\ref{alg:example-non-trivial-rsm}]
    \label{example:past-of-example-program}
   There exists an inductive bound with $\epsilon\leq0.1128$ and  $d\geq0.4102$ for $p=0.5$. These constants are chosen so that convergence is guaranteed; see  Appendix~\ref{appendix:example-inductive-bound-for-linear} for details. As $d=\delta'k^{2\deg(\mathcal{P})+1}-1$ and $\delta'$ converges to $1$ (Lemma~\ref{lemma:variance-growth-polynomial-random-walk}), $k\geq 1.1214$ ensures that $d$ is large enough when the degree of a polynomial random walk program is at least $1$ (that is, at least linear updates). 

    Through this bound and Lemma~\ref{lemma:bound-from-termination-probability}, the stopping time of a polynomial random walk program $\mathcal{P}$ with $\deg(\mathcal{P})\geq 1$ is bounded:  $\mathbb{E}(T)\leq \sum_{n=0}^{\infty} Bn^{\frac{\ln(0.8872)}{\ln{1.1214+\tau}}}$. This stopping time bound has an exponent which is smaller than $-1.04$;  therefore, the loop summation of the respective  Figure~\ref{alg:random_walk_transformed}
    is finite and  $d_\text{min}(0.5)\leq 1$.

    Using Theorem~\ref{theorem:polynomial-random-walk-threshold} we conclude that  polynomial random walks with linearly (or faster) increasing step size and branching probability $0.5$ have finite expected stopping time and are PAST, given that $\deg(\mathcal{P})>\deg(p(q_1[n])+(1-p)q_2[n])$. In particular, this is true for Figure~\ref{alg:example-non-trivial-rsm}, as shown in Example~\ref{ex:RSM:bounded}, hence it is PAST. \qed
    
\end{example}

\paragraph{\bf Higher moments of the stopping time}\label{subsection:higher-moments} 
We conclude this section by noting that solving~\eqref{equation:optimization-problem-random-walk} and applying Theorem~\ref{theorem:polynomial-random-walk-threshold} allows us to derive not only PAST, but also higher moments of the stopping times of polynomial random walks $\mathcal{P}$. That is, the bound we compute for $\mathbb{P}(T\geq n)$ by solving~\eqref{equation:optimization-problem-random-walk} is of the form $Bn^m$ and this bound can be used to bound higher moments $N$ of the stopping time. In particular, $\mathbb{E}(T^N) =\sum_{n=0}^\infty \mathbb{P}(T^N \geq n) = \sum_{n=0}^\infty \mathbb{P}(T \geq \sqrt[\leftroot{-2}\uproot{2}N]{n}) \leq \sum_{n=0}^\infty Bn^\frac{m}{N}$. Therefore, when~\eqref{equation:optimization-problem-random-walk} is solved using a bound with $m<-N$, then $\mathbb{E}(T^N)$ is finite.

\section{Implementation and Experiments}
\label{section:implementation}
Theorem~\ref{theorem:polynomial-random-walk-threshold} states sufficient conditions under which the polynomial random walk programs $\mathcal{P}$ of Figure~\ref{alg:random_walk} are PAST. 
These sufficient conditions can be checked by solving inequalities among random walk updates and, importantly, by  finding solutions to  the optimization problem of~\eqref{equation:optimization-problem-random-walk}.  
In this section,  we detail our implementation to find tight bounds
to~\eqref{equation:optimization-problem-random-walk}, allowing us to conclude PAST of $\mathcal{P}$.
Our implementation involves heuristic optimization techniques to find provably correct solutions. Our experiments provide practical evidence on the tightness of computed stopping times bounds and give evidence of the reliability of our approach\footnote{see Appendix~\ref{appendix:performance-genetic-algorithm}}, despite the absence of convergence guarantees. 

\subsection{Computing Tight Bounds on Stopping Times}

We  solve~\eqref{equation:optimization-problem-random-walk} in extension of the \verb|Polar| program analyzer~\cite{DBLP:journals/pacmpl/MoosbruggerSBK22}. We use \verb|Polar| to compute closed form expressions for the loop-guard changes of  probabilistic branches, allowing us to support programs $\mathcal{P}$ that are even more general than 
Figure~\ref{alg:random_walk}.
We combine \verb|Polar| with linear programming through~\verb|OR-Tools|~\cite{ortools} and derive inductive bounds for fixed program transformation parameters. To find the best values for these parameters, we rely on 
genetic algorithms, such that the fitness functions of these genetic algorithms are  controlled by our linear solver. Doing so, we use the  \verb|Gurobi|-solver~\cite{gurobi} to solve  linear models.  
By integrating algebraic reasoning, linear programming and genetic algorithms, our implementation in \verb|Polar| minimizes the exponent in the bound of $\mathbb{P}(T\geq n)$ in~\eqref{equation:optimization-problem-random-walk}, which is sufficient to prove PAST and finiteness of further higher moments of $\mathcal{P}$ (Theorem~\ref{theorem:polynomial-random-walk-threshold}). 
By changing the objective function,
 our implementation can also 
minimize an \emph{explicit} bound for the expected stopping time $\mathbb{E}(T)$.


\paragraph{\bf Inferring inductive bound sets.}
\label{subsection:computing-inductive-bound-set}
To compute an inductive bound set $B$ in \eqref{eq:inductiveBoundSet}, the parameters $\epsilon, d,c_0, C_1, \delta_1$ must be fixed. Additionally, we require vectors $\vec{a}_B$ and $\vec{c}_B$, specifying respectively which $m$ lower bounds for the inductive bound-set are computed and which $k$ tail bounds are used. 
We compute values for the bounds by solving the linear inequality: 
{\small 
\begin{equation}\label{equation:inductiv-bound-inductivity-inequality}
    \begin{array}{lcl}
    \vec{b}_{i} & \geq &  \sum_{j=1}^{m} \vec{d}_j \left(\Phi\left(\frac{\vec{a}_i\sqrt{1+d}-\vec{a}_j}{\sqrt{d}}\right)-c_0\right) + \sum_{j=1}^{k} \vec{d}_{m+j}\left(\Phi\left(\frac{\vec{a}_i\sqrt{1+d}-\vec{c}_j}{\sqrt{d}}\right)-c_0\right)~~\hspace*{-1em}
\end{array}
\end{equation}
}
\noindent for $i=1, \ldots, m$.  
The vector $\vec{d}$ denotes auxiliary variables, which describe the difference of neighbouring bounds. Additionally, we enforce\footnote{see  Appendix~\ref{appendix:details-inductive-bound-model} for details} that the initial, almost Normal, distribution of $Z_0$ satisfies the bound set $B$. 

Our implementation invokes linear programming over the linear model~\eqref{equation:inductiv-bound-inductivity-inequality} in the form of an indicator function. This function returns $1$, when an inductive bound set $B$ is found for the given parameters $\epsilon, d,c_0, C_1, \delta_1$; and $0$ otherwise. 


\paragraph{\bf Genetic algorithm.} 
We use a genetic algorithm to solve the optimization problem~\eqref{equation:optimization-problem-random-walk} and find the best parameter values in~\eqref{equation:inductiv-bound-inductivity-inequality}, for which an inductive bound set $B$ exists. 
Our genetic algorithm repeatedly modifies a collection of individual solutions: we select individuals from the current set of solutions and use them to produce next individuals/solutions. 
An individual has (i) the properties $d$, $\epsilon$,  and $n_0$ to capture the program transformation of Figure~\ref{alg:random_walk_transformed} and (ii) the parameters $g$, $s$, $c$ to specify the vectors $\vec{a}$ and $\vec{c}$ of~\eqref{equation:inductiv-bound-inductivity-inequality} for the inductive bound $B$. Specifically, we set $\vec{a}_1 =0,\vec{a}_2=\frac{s}{g-1},\dots,\vec{a}_g=s$, and $\vec{c}=\begin{pmatrix}c\end{pmatrix}$.

The fitness of an individual is calculated by first calculating the exponent of the bound. In case an explicit bound should be computed, the error-terms $c_0$ and $\delta_1$ are inferred from $n_0$. Otherwise, we choose very small values, such as  $c_0 = \delta_1=10^{-8}$, $\delta'=1+10^{-8}$. Next, we solve  our linear model~\eqref{equation:inductiv-bound-inductivity-inequality}. If no solution is found, we set $m=0$;  otherwise, we take  $m=\frac{\ln(1-\epsilon)}{\ln(k+\frac{1}{n_0})}$ with $k = \left(\frac{(d+1)}{\delta'}\right)^{\frac{1}{2\deg(\mathcal{P})+1}}$. If  $m<-1$ and an explicit bound is sought, we compute the summation $\mathbb{E}(T)=\sum_{n=1}^{\infty}\mathbb{P}(T\geq n)$ using the Hurwitz $\zeta$-function~\footnote{\href{https://dlmf.nist.gov/25.11.i}{(25.11)}{NIST:DLMF}}. The fitness of an individual is further expressed via the tuple $(\mathbb{E}(T),m)$, which is 
minimized/optimized with respect to the usual lexicographical ordering.

From a given solution set (generation), a new solution set (population) is generated using random mutations of the parameters. The property $d$ is biased to decrease, while $\epsilon$ is biased to increase. Additionally, $n_0$ is biased to increase when the exponent is not smaller than $-1$, and biased to decrease otherwise. Furthermore, new individuals are generated by randomly selecting the properties of two parent individuals.

\subsection{Experimental Results}
\label{subsection:experiments}

We evaluated our approach for computing stopping time bounds, and hence, inferring PAST,  using  various polynomial random walk programs $\mathcal{P}$. To this end, we took instances of Figure~\ref{alg:random_walk} with different random walk degrees $\deg{P}$ and various values of the probabilistic choice $p$. 
The PAST analysis of such programs is out of reach of existing tools (see Section~\ref{section:related}), notably
\verb|Amber|~\cite{moosbrugger2021probabilistic},
\verb|eco-imp|~\cite{ecoimp},
\verb|KoAT|~\cite{koat},  \verb|LexRsm|~\cite{agrawal2017lexicographicrankingsupermartingalesefficient}, and  \verb|LazyLexRsm|~\cite{takisaka-2024-lex-rsm-lower-bound}.

Table~\ref{table:empirical-survival-rate-exponent} summarizes our experiments, with the third line of  Table~\ref{table:empirical-survival-rate-exponent} being our motivating example from Figure~\ref{alg:example-non-trivial-rsm}.  Column~3 of Table~\ref{table:empirical-survival-rate-exponent} reports the empirical exponents of $\mathbb{P}(T\leq n)$, further detailed in Figure~\ref{fig:empirical-degree}. Column~4 of Table~\ref{table:empirical-survival-rate-exponent} states the smallest (tightest) exponent obtained through our approach using inductive bounds. Our experiments were run on a machine with 2x AMD EPYC 7502 32-Core processor with one task per core and hyperthreading disabled.

\begin{table}[t]
\caption{Derived bounds on stopping times for polynomial random walk programs $\mathcal{P}$, with increasing maximal degree $\deg{P}$ and different probabilistic choices $p$. The program in the 3rd line of the table corresponds to Figure~\ref{alg:example-non-trivial-rsm}. }\label{table:empirical-survival-rate-exponent}
\centering
\begin{tabular}{|l|l|l|l|}
\hline
$\deg(\mathcal{P})$ & $p$&  measured exponent & tightest bound \\
\hline
$0.25$ & $0.5$& $-0.744$ & $-0.5589$ \\
$0.5$&$0.5$&$-0.997$ & $-0.7436$	\\
$1$&$0.5$& $-1.508$ & $-1.1189$ \\
$2$&$0.5$&$-2.442$ & $-1.8639$ \\
$5$&$0.5$&$-4.588$ & $-4.0843$ \\

$3$ & $0.5$ & $-3.448$ & $-2.5971$\\
$3$ & $0.9$ & $-3.334$ & $-2.4453$\\
$3$ & $0.1$ & $-3.57$ & $-2.4453$\\

$3$ & $0.99$ & $-3.152$ & $-1.9321$\\
$3$ & $0.01$ & $-3.516$ & $-1.9321$\\

$3$ & $0.999$ & $-3.144$ & $-1.359$\\
$3$ & $0.001$ & $-3.562$ & $-1.359$\\
\hline
\end{tabular}
\end{table}

\paragraph{\bf Experimental analysis.} Figure~\ref{fig:empirical-degree} displays empirically measured  rates of $\mathbb{P}(T\geq n)$  for symmetric random walks with varying degree. These probabilities  appear to converge towards a line in the log-log plot, which suggests, that $\mathbb{P}(T\geq n)$ eventually is of form $Bn^m$, coinciding with the form of our bound. The observed exponent of this probability is the slope of the robust log-log regression lines~\cite{rlm}, displayed as dashed lines and displayed in Column~3 of Table~\ref{table:empirical-survival-rate-exponent}.


\begin{figure}
\subfloat[with different degrees\label{fig:empirical-degree}]{%
  \includegraphics[width=0.5\textwidth]{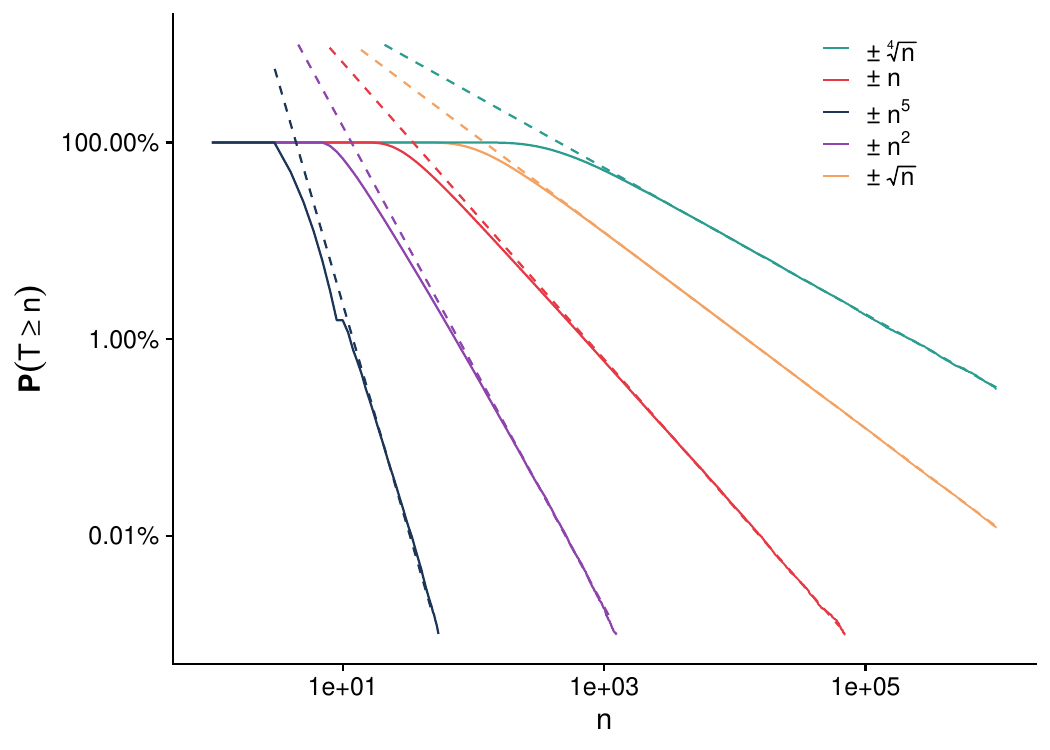}%
}\hfil
\subfloat[with different percentages\label{fig:empirical-different-percentages}]{%
  \includegraphics[width=0.5\textwidth]{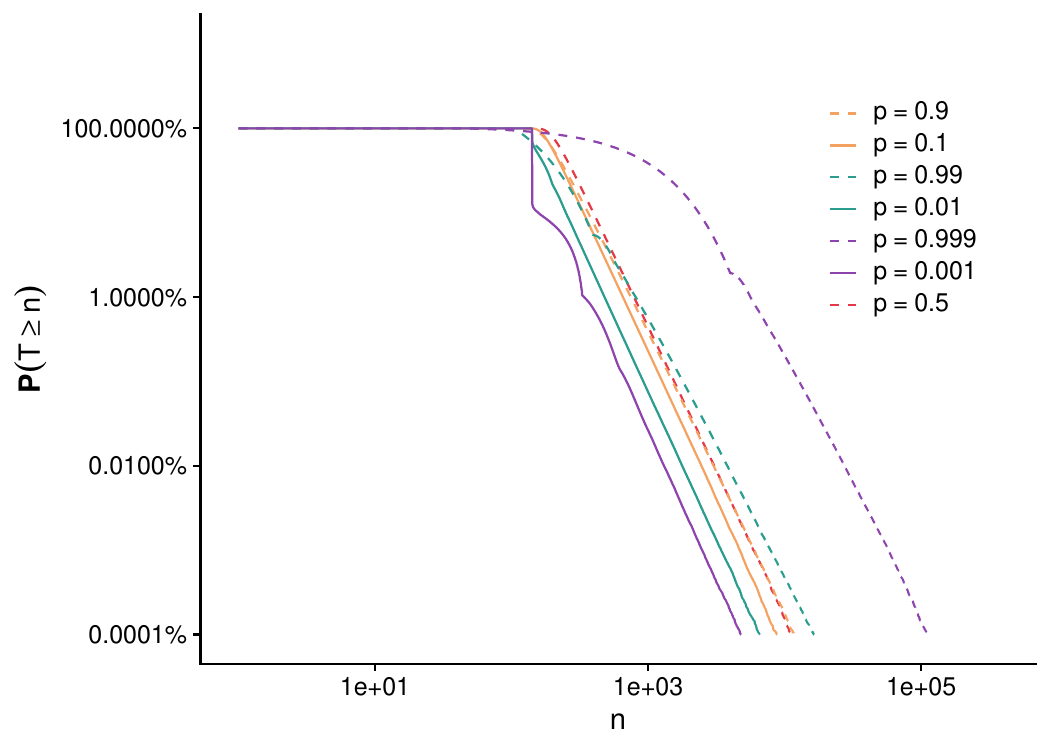}%
}\hfil

\caption{Empirical results on stopping times of polynomial random walks}
\label{fig:empirical-survival-rates}
\end{figure}


In Figure~\ref{fig:empirical-different-percentages} we display the stopping times $\mathbb{P}(T\geq n)$ of zero-mean polynomial random walks with different values of $p$ and degree $3$.
The approximated values of the exponent, as well as the tightest bound found by our method can be found in Table~\ref{table:empirical-survival-rate-exponent}. The increasing unsharpness for small (or large) values of $p$ stem from the use of Hoeffding's lemma in the proof of Lemma~\ref{lemma:tail-bound}. While for individual bounds of centered $X_i$ this bound is sharp, for the product this no longer is the case and the plot suggests, that a tighter bound might be found.

\subsubsection{Explicit bound analysis.}
\begin{wrapstuff}[r, width=0.5\textwidth]
{
\setlength{\intextsep}{0pt}
\begin{figure}[H]
\centering
  \includegraphics[width=0.96\textwidth]{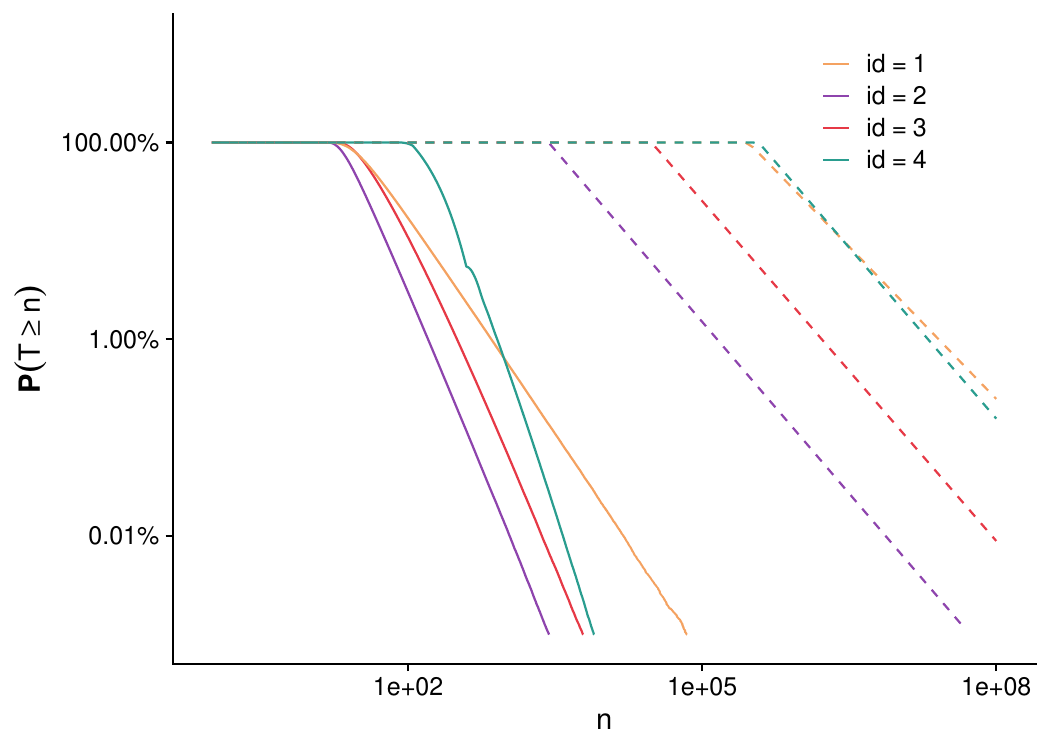}
  \caption{Examples of explicit bounds}
  \label{fig:explicit-bounds}
\end{figure}
}
\end{wrapstuff}
Our  genetic algorithms can be used to  compute  explicit bounds on the running times. Figure~\ref{fig:explicit-bounds} shows the tightest explicit bound found for some random walk programs. 
The explicit bound is off by several orders of magnitude, as listed in Table~\ref{table:explicit-bounds}.
One of the main reasons for the explicit bound being unsharp stems from the fact that we choose a specific $n_0$ in~\eqref{equation:optimization-problem-random-walk}, from which the parameters for the inductive bounds $B$ are computed. This could be improved by computing a bound with multiple segments, and therefore inferring multiple exponents that decrease with growing $n$.




\begin{table}[t]
\caption{Explicit bounds and empirical means of stopping time, with various polynomial updates $q_1,q_2$ and probabilistic choice $p$ of polynomial random walks from Figure~\ref{alg:random_walk}  }\label{table:explicit-bounds}
\centering
\begin{tabular}{|l|l|l|l|l|l|l|}
\hline
id&$q_1$ & $p$& $q_2$& $y_0$ & empiric $\mathbb{E}(T)$ & explicit bound \\
\hline
$1$&$n$ & $0.5$& $-n$ & $100$ &$89.62$& $9562887$\\
$2$&$n^2$&$0.5$&$-n^2$ & $100$ & $36.4$ &$17708$	\\
$3$&$n^2+2n+20$&$0.5$& $-n^2+2n+20$ & $1000$ & $59.47$&$213570$ \\
$4$&$\frac{n^3}{0.99}$&$0.99$&$-\frac{n^3}{0.01}$ & $10^8$ &$212.8$& $2671328$ \\
\hline
\end{tabular}
\end{table}

\section{Related Work}
\label{section:related}

Reasoning about probabilistic program termination is much harder than for deterministic programs~\cite{hark-aiming-2020}, turning the automated analysis of probabilistic loops into a challenging problem. Most approaches rely on proof rules for proving (positive) almost-sure termination~\cite{majumdar-past-complexity-proof-rules,chakarov:supermartingale-ranking-functions,chatterjee_stochastic_invariants,meyer-expected-runtimes-probabilistic-integer-programs}, which in turn require additional expressions, notably loop invariants and martingale variants, for the applicability in the proof rules. As the generation of invariants and martingales is  undecidable in general, automation of these approaches requires user-provided invariants/martingales. Our work is limited to polynomial random walks with the benefit of providing sufficient conditions under which PAST can automatically be inferred.  While restrictive, Figure~\ref{alg:example-non-trivial-rsm} shows advantages of our approach: proving PAST of this program using proof rules from~\cite{majumdar-past-complexity-proof-rules,chakarov:supermartingale-ranking-functions,chatterjee_stochastic_invariants} requires an auxiliary ranking super-martingale, whose computability is still an open question.

Alternative approaches to automating termination analysis have been proposed by focusing on restrictive classes of probabilistic loops, whose (P)AST analysis becomes (semi-)decidable~\cite{giesl:constant-probability-constant-updates,moosbrugger2021probabilistic}. Notably,  constant probability loops~\cite{giesl:constant-probability-constant-updates} limit probabilistic loop updates to constant increments over random  variable and their (P)AST is decidable. 
A more expressive class of programs is given in \cite{moosbrugger2021probabilistic}, with (P)AST analysis shown to be semi-decidable and automated. Key to automation is the ability to inferring (ranking) super-martingales from loop guards and relaxing proof rules to ``eventual'' reasoning over polynomial loop updates.  Our approach complements these works by using arbitrary polynomial updates in 
polynomial random walks.  Such loops cannot be analyzed by~\cite{giesl:constant-probability-constant-updates,moosbrugger2021probabilistic}; in particular, PAST of Figure~\ref{alg:example-non-trivial-rsm}
cannot be inferred. 

The analysis of probabilistic programs with arbitrary polynomial updates and control flow is shown to be  difficult, especially due to the lack of compositionality~\cite{kaminski:2018-compositionality}. 
By adjustments of the weakest precondition~\cite{mciver-new-2017,McIver2005} calculus, runtime bounds are inferred as  sufficient conditions for proving PAST in~\cite{absynth,ecoimp}. Control-flow refinement methods are also advocated in~\cite{koat} to derive runtime bounds on probabilistic loops. 
Further, 
lexicographical extensions of synthesizing ranking super-martingales are presented in~\cite{agrawal2017lexicographicrankingsupermartingalesefficient,takisaka-2024-lex-rsm-lower-bound} for the purpose of PAST inference. While powerful, automation of these works depend on the suitable martingales. Unlike our technique,  proving PAST of polynomial random walks, in particular of Figure~\ref{alg:example-non-trivial-rsm}, cannot yet be achieved by other works.

\section{Conclusions}
We study the positive almost-sure termination (PAST) problem of polynomial probabilistic programs implementing random walks with increasing increments. 
We show that PAST can be proven for polynomial random walks by checking conditions via solving linear inequalities over the polynomial program updates, without requiring additional user input in the form of invariants and/or martingales.
Our experiments demonstrate that our approach determines PAST of non-trivial probabilistic programs. 
Notably, we show PAST for programs 
beyond the scope of existing methods: for such programs, state-of-the-art works would require ranking super-martingales whose computation  is undecidable in general. For such loops, we prove PAST by finding bounds on the probability of termination, depending on the degrees and the branching probability of the polynomial updates.
Future work includes the extension of our results to (i) deriving hardness results on PAST decidability for polynomial random walks, and (ii) dealing with probabilistic programs with nondeterminism and more complex updates. 

\begin{credits}
\subsubsection{\ackname} This research was partially supported by 
the European Research Council Consolidator Grant ARTIST 101002685 and  the 
Vienna Science and Technology Fund WWTF 10.47379/ICT19018 grant ProbInG. 
We thank Marcel Moosbrugger for valuable discussions and  providing guidance on the Polar analyzer.

\subsubsection{Disclosure of Interests.}
The authors have no competing interests to declare that are relevant to the content of this article.

\subsubsection{Data Availability.} The models, scripts, and tools to reproduce our experimental evaluation are archived and publicly available at DOI~\href{https://doi.org/10.5281/zenodo.15257958}{10.5281/zenodo.15257958}.

\end{credits}
%
%
%
\bibliographystyle{splncs04}
\bibliography{variance_based_past}

\appendix

\newcommand{\LongProof}[1]{}

\clearpage\section{Appendix}
We present detailed proofs of our main results. 

\subsection{Correctness of Properties of Almost Normal Variables and Conditioning }\label{appendix:section2}

We prove the main properties of Section~\ref{section:summation-almost-normal}.

\lemmaUpperBoundSnRightTail*

    \begin{proof}
        We prove by induction over $n$ in $S_n$. 
        
        \noindent \emph{Base case.} By definition, $S_0 = Z_0$.  Lemma~\ref{definition:z_n} bounds $Z_0$ and we conclude:
        \[\begin{array}{lcl}\mathbb{P}\bigg(S_0 \geq a + (\frac{\sqrt{1+d}}{\sqrt{1+d}-1}\delta_1+b)\sigma_0 ~\bigg)& \leq& \exp\left(C_1\frac{-(a+(\frac{\sqrt{1+d}}{\sqrt{1+d}-1}\delta_1+b)\sigma_0)^2}{2\sigma_0^2}\right) \\
        & \leq & \exp\left(C_1\frac{-a^2}{2\sigma_0^2}\right)\end{array}\]
        \emph{Induction step}.
        \noindent Assume: 
        \[\begin{array}{lll}
        \mathbb{P}\bigg(S_n \geq a + (\frac{\sqrt{1+d}}{\sqrt{1+d}-1}\delta_1+b) \sqrt{\sum_{i=0}^n\sigma_i^2}~\bigg)&\leq&
        \exp\left(C_1\frac{-a^2}{2\sum_{i=0}^n\sigma_i^2}\right)\end{array}.\] 
       By conditioning, the probability mass of $\epsilon$ is cut away. Therefore, the tail probability of $S_n$ increases by $\frac{1}{1-\epsilon}$ at most and we have: 
       
        $\begin{array}{lll}
        \mathbb{P}\bigg((S_n | S_{n}>F_{S_n}^{-1}(\epsilon) & \geq & a + (\frac{\sqrt{1+d}}{\sqrt{1+d}-1}\delta_1+b)\sqrt{\sum_{i=0}^n\sigma_i^2}\bigg)\\
        &\leq& \frac{1}{1-\epsilon}\exp\left(C_1\frac{-a^2}{2\sum_{i=0}^n\sigma_i^2}\right) \leq \exp\left(C_1\left(\frac{-a^2}{2\sum_{i=0}^n\sigma_i^2}\right) + \ln{\frac{1}{1-\epsilon}}\right).\end{array}$
        
    \noindent
       Using the  bound of $\sum_{i=0}^{n+1}\sigma_i^2$ together with algebraic simplifications, we get: 
        
        
        $\begin{array}{l}
        \mathbb{P}\left((S_n | S_{n}>F_{S_n}^{-1}(\epsilon)) \geq a + b\sqrt{\sum\nolimits_{i=0}^{n+1}\sigma_i^2} + \frac{\sqrt{1+d}}{\sqrt{1+d}-1}\delta_1 \sqrt{\sum\nolimits_{i=0}^{n}\sigma_i^2}\right)\\
        \leq
\exp\left(C_1\left(\frac{-a^2}{2\sum_{i=0}^n\sigma_i^2}\right)\right).\end{array}$\\

\noindent        
The sum of sub-Gaussian variables bounded by variances $\sigma_1^2$ and $\sigma_2^2$ admits a bound with variance $\sigma_1^2 +\sigma_2^2$~\cite{lattimore2020bandit}. 
     \\ The random variable $\left((S_{n}|S_{n}>F_{S_n}^{-1}(\epsilon) ) -(b+\frac{\sqrt{1+d}}{\sqrt{1+d}-1}\delta_1)\sqrt{\sum_{i=0}^{n+1}\sigma_i^2}\right)$ is thus bounded, while  $(Z_{n+1}-\delta_1\sigma_{n+1})$ is bounded by Lemma~\ref{definition:z_n}. Hence, \\
%
 %
$\begin{array}{l}
            \mathbb{P}\Bigg(\left((S_{n}|S_{n}>F_{S_n}^{-1}(\epsilon) ) -b\sqrt{\sum\nolimits_{i=0}^{n+1}\sigma_i^2} - \frac{\sqrt{1+d}}{\sqrt{1+d}-1}\delta_1\sqrt{\sum\nolimits_{i=0}^{n}\sigma_i^2}\right)\\\qquad+(Z_{n+1}-\delta_1\sigma_{n+1}) \geq a\Bigg)~~~\leq ~~~\exp\left(C_1\left(\frac{-a^2}{2\sum_{i=0}^{n+1}\sigma_i^2}\right)\right)
        \end{array}$.

\noindent By the growth of the variance, we have $\sqrt{\sum_{i=0}^{n}\sigma_i^2}\leq 
 \sqrt{\frac{\sigma_{n+1}^2}{1+d}} = \frac{\sigma_{n+1}}{\sqrt{1+d}}$. Also,  $\sigma_{n+1}\leq\sqrt{\sum_{i=0}^{n+1}\sigma_i^2}$. Putting it all together, the induction step follows: 
 
%
$\begin{array}{l}
 \mathbb{P}\left(\left((S_{n}|S_{n}>F_{S_n}^{-1}(\epsilon) ) + Z_{n+1}\geq a+(b+\frac{\sqrt{1+d}}{\sqrt{1+d}-1}\delta_1)\sqrt{\sum_{i=0}^{n+1}\sigma_i^2}\right)\right)\\\qquad\leq \exp\left(C_1\left(\frac{-a^2}{2\sum_{i=0}^{n+1}\sigma_i^2}\right)\right)
\end{array}$.

    \qed
    \end{proof}

\LongProof{    
    \begin{proof}
        By induction. \emph{Base case.} $S_0 = Z_0$, and the bound of $Z_0$ is stronger, since $\delta_1$ and $b$ ar positive:
    
        \begin{multline}\mathbb{P}(S_0 \geq a + (\frac{\sqrt{1+d}}{\sqrt{1+d}-1}\delta_1+b)\sigma_0)\leq \exp\left(C_1\frac{-(a+(\frac{\sqrt{1+d}}{\sqrt{1+d}-1}\delta_1+b)\sigma_0)^2}{2\sigma_0^2}\right) \\\leq \exp\left(C_1\frac{-a^2}{2\sigma_0^2}\right)\end{multline}
    
        \emph{Induction step}.
        \noindent We know that \begin{align}
        \mathbb{P}(S_n \geq a + (\frac{\sqrt{1+d}}{\sqrt{1+d}-1}\delta_1+b) \sqrt{\sum_{i=0}^n\sigma_i^2})\leq \exp\left(C_1\frac{-a^2}{2\sum_{i=0}^n\sigma_i^2}\right)\end{align}
        When cutting away a fraction $\epsilon$ of the left tail, we multiply the tail probability by $\frac{1}{1-\epsilon}$:
        \begin{multline}\mathbb{P}((S_n | S_{n}>F_{S_n}^{-1}(\epsilon)) \geq a + (\frac{\sqrt{1+d}}{\sqrt{1+d}-1}\delta_1+b)\sqrt{\sum_{i=0}^n\sigma_i^2})\\\leq \frac{1}{1-\epsilon}\exp\left(C_1\frac{-a^2}{2\sum_{i=0}^n\sigma_i^2}\right) \leq \exp\left(C_1\left(\frac{-a^2}{2\sum_{i=0}^n\sigma_i^2}\right) + \ln{\frac{1}{1-\epsilon}}\right).\end{multline}
    
        Now inserting the desired bound involving $\sum_{i=0}^{n+1}\sigma_i^2$:

        \begin{multline}\mathbb{P}\left((S_n | S_{n}>F_{S_n}^{-1}(\epsilon)) \geq a + b\sqrt{\sum_{i=0}^{n+1}\sigma_i^2} + \frac{\sqrt{1+d}}{\sqrt{1+d}-1}\delta_1 \sqrt{\sum_{i=0}^{n}\sigma_i^2}\right)\\\leq \exp\left(C_1\left(\frac{-\left(a+b\left(\sqrt{\sum_{i=0}^{n+1}\sigma_i^2}-\sqrt{\sum_{i=0}^{n}\sigma_i^2}\right)\right)^2}{2\sum_{i=0}^n\sigma_i^2}\right)  + \ln{\frac{1}{1-\epsilon}}\right).\end{multline}
    
        We want to show, that:
        \begin{multline}
\exp\left(C_1\left(\frac{-\left(a+b\left(\sqrt{\sum_{i=0}^{n+1}\sigma_i^2}-\sqrt{\sum_{i=0}^{n}\sigma_i^2}\right)\right)^2}{2\sum_{i=0}^n\sigma_i^2}\right)  + \ln{\frac{1}{1-\epsilon}}\right)\\\leq
\exp\left(C_1\left(\frac{-a^2}{2\sum_{i=0}^n\sigma_i^2}\right)\right)
\end{multline}
        holds, which is the case, when
        \begin{align}
            -\left(a+b\left(\sqrt{\sum\nolimits_{i=0}^{n+1}\sigma_i^2}-\sqrt{\sum\nolimits_{i=0}^{n}\sigma_i^2}\right)\right)^2+ \frac{2\left(\sum\nolimits_{i=0}^n\sigma_i^2\right)}{C_1} \ln{\frac{1}{1-\epsilon}} \leq -a^2
        \end{align}
        By inserting the inequality $\sigma_{n+1}^2\geq d \sum_{i=0}^n\sigma_i^2$, hence $\sum_{i=0}^{n+1}\sigma_i^2\geq (d+1) \sum_{i=0}^n\sigma_i^2$:

        \begin{multline}
            -\left(a+b\left(\sqrt{(1+d)\sum\nolimits_{i=0}^{n}\sigma_i^2}-\sqrt{\sum\nolimits_{i=0}^{n}\sigma_i^2}\right)\right)^2+ \frac{2\left(\sum\nolimits_{i=0}^n\sigma_i^2\right)}{C_1} \ln{\frac{1}{1-\epsilon}} 
            \\
            =-\left(a+b(\sqrt{(1+d)}-1)\sqrt{\sum\nolimits_{i=0}^{n}\sigma_i^2}\right)^2+ \frac{2\left(\sum\nolimits_{i=0}^n\sigma_i^2\right)}{C_1} \ln{\frac{1}{1-\epsilon}} 
            \leq -a^2
        \end{multline}
    
        which again is true, when
        \begin{align}
            -b^2\left(\sqrt{(1+d)}-1\right)^2\sum\nolimits_{i=0}^{n}\sigma_i^2 + \frac{2\left(\sum\nolimits_{i=0}^n\sigma_i^2\right)}{C_1} \ln{\frac{1}{1-\epsilon}} 
            \leq 0
        \end{align}
        and equivalently when the inequality for $b$ holds
        \begin{align}
            b^2\left(\sqrt{(1+d)}-1\right)^2\geq\frac{2}{C_1} \ln{\frac{1}{1-\epsilon}}\Leftrightarrow b \geq\frac{ \sqrt{2 \ln{\frac{1}{1-\epsilon}}}}{\sqrt{C_1}(\sqrt{(1+d)}-1)}.
        \end{align}
        As this is true by the precondition of the lemma, we can conclude, that: 
            \begin{multline}
                \mathbb{P}\left((S_n | S_n >F_{S_n}^{-1}(\epsilon)) \geq a + b\sqrt{\sum_{i=0}^{n+1}\sigma_i^2}+\frac{\sqrt{1+d}}{\sqrt{1+d}-1}\delta_1\sqrt{\sum_{i=0}^{n}\sigma_i^2}\right)\\\leq \exp\left(C_1\left(\frac{-a^2}{2\sum_{i=0}^n\sigma_i^2}\right)\right)
            \end{multline}

        When two random variables are $\sigma_1$ and $\sigma_2$ sub-gaussian, then their sum is $\sqrt{\sigma_1^2 +\sigma_2^2}$ sub-gaussian \cite[Lemma~5.4c]{lattimore2020bandit}.
    
        For the random variable $\left((S_{n}|S_{n}>F_{S_n}^{-1}(\epsilon) ) -(b+\frac{\sqrt{1+d}}{\sqrt{1+d}-1}\delta_1)\sqrt{\sum_{i=0}^{n+1}\sigma_i^2}\right)$ we have just established a bound, while for $(Z_{n+1}-\delta_1\sigma_{n+1})$ we have a bound from the precondition of the theorem.
    
        Observe, that we can multiply the factor $C_1$ into the variance-proxy, resulting in $\sqrt{\frac{1}{C_1}\sum_{i=0}^{n}\sigma_i^2}$ and $\sqrt{\frac{1}{C_1}} \sigma_{n+1}$ sub-gaussian random variables.
        
        Through applying the just mentioned Lemma, we get the following bound:
        \begin{multline}
            \mathbb{P}\left(\left((S_{n}|S_{n}>F_{S_n}^{-1}(\epsilon) ) -b\sqrt{\sum_{i=0}^{n+1}\sigma_i^2} - \frac{\sqrt{1+d}}{\sqrt{1+d}-1}\delta_1\sqrt{\sum_{i=0}^{n}\sigma_i^2}\right)+(Z_{n+1}-\delta_1\sigma_{n+1}) \geq a\right)\\\leq \exp\left(C_1\left(\frac{-a^2}{2\sum_{i=0}^{n+1}\sigma_i^2}\right)\right)
        \end{multline}

By the growth of the variance, it holds that $\sqrt{\sum_{i=0}^{n}\sigma_i^2}\leq 
 \sqrt{\frac{\sigma_{n+1}^2}{1+d}} = \frac{\sigma_{n+1}}{\sqrt{1+d}}$. Also, trivially, $\sigma_{n+1}\leq\sqrt{\sum_{i=0}^{n+1}\sigma_i^2}$. Using those inequalities, and the fact that
\begin{align}
\delta_1\sigma_{n+1}+\frac{\sqrt{1+d}}{\sqrt{1+d}-1}\delta_1\frac{\sigma_{n+1}}{\sqrt{1+d}}=\frac{\sqrt{1+d}}{\sqrt{1+d}-1}\delta_1\sigma_{n+1}
\end{align}
since, when cancelling fractions and multiplying by the common denominator,
\begin{align}
\delta_1\sigma_{n+1}(\sqrt{1+d}-1)+\delta_1\sigma_{n+1}=\sqrt{1+d}\delta_1\sigma_{n+1}
\end{align}
we get that
 \begin{align}
    \delta_1\sigma_{n+1}+\frac{\sqrt{1+d}}{\sqrt{1+d}-1}\delta_1\sqrt{\sum_{i=0}^{n}\sigma_i^2}\leq \frac{\sqrt{1+d}}{\sqrt{1+d}-1}\delta_1\sqrt{\sum_{i=0}^{n+1}\sigma_i^2}.
\end{align}

Therefore, we can show the induction step:
\begin{multline}
    \mathbb{P}\left(\left((S_{n}|S_{n}>F_{S_n}^{-1}(\epsilon) ) + Z_{n+1}\geq a+(b+\frac{\sqrt{1+d}}{\sqrt{1+d}-1}\delta_1)\sqrt{\sum_{i=0}^{n+1}\sigma_i^2}\right)\right)\\\leq \exp\left(C_1\left(\frac{-a^2}{2\sum_{i=0}^{n+1}\sigma_i^2}\right)\right)
\end{multline}
    \qed
    \end{proof}
}

\bigskip

\theoremEpsilonExistsAlways*
\begin{proof}
Lemma~\ref{lemma_upper_bound_s_n_right_tail} implies that $\mathbb{P}(S_n\geq c)$ is bounded and converges to $0$ as $c\rightarrow\infty$. We choose an arbitrary $c'$ such that  the tail bound $\mathbb{P}(S_n\geq c')\leq b'<\frac{1}{2}$ holds. 
Then, $\mathbb{P}(S_n< c')\geq1-b'$ and $\mathbb{P}\big((S_n | S_n\geq 0)< c'\big)\geq1-b'-\epsilon\geq 0$. 
By union-bound properties, 
we have $\mathbb{P}\bigg(S_{n+1}\leq0\bigg) \geq\mathbb{P}\bigg((S_n | S_n\geq 0)< c'\bigg)\cdot\allowbreak\mathbb{P}\bigg(Z_{n+1}<-c'\bigg)\geq (1-b'-\epsilon)(\Phi(\frac{c'}{\sqrt{d}})-c_0)$. The bound with $\vec{a}=\begin{pmatrix}0\end{pmatrix}$ and $\vec{b}=\begin{pmatrix}\frac{(1-b')(\Phi(\frac{c'}{\sqrt{d}})-c_0)}{1+(\Phi(\frac{c'}{\sqrt{d}})-c_0)})\end{pmatrix}$ is then inductive and the lower bound is nonzero when $c_0$ is small enough.
    \qed
\end{proof}

\subsection{Correctness of Polynomial Random Walks Properties}

Here we prove the key results of Section~\ref{section:polynomial-random-walks}. 

\lemmaStoppingTimeInequality*
\begin{proof}
    For every possible draw of the variables $X_n$, the termination of Figure~\ref{alg:random_walk_transformed} implies the termination of Figure~\ref{alg:random_walk}

    In Figure~\ref{alg:random_walk_transformed}, the inner loop  is executed at least once per execution of the outer loop's body, hence $n$ is incremented at least once per iteration of the outer loop. When the loop stops after $t$ steps, then either Figure~\ref{alg:random_walk} also stops after $t$ steps or it has already stopped at some $t' < t$.
    \qed
\end{proof}


\bigskip

\lemmaSummationCdfDeviation*

\LongProof{
\begin{proof}
        We first consider the centered version of $U_{n'}$: \[U_{n'}' = X_{n'} - \mathbb{E}(X_{n'}) + \dots + X_{\lceil n'\cdot k\rceil} - \mathbb{E}(X_{\lceil n'\cdot k\rceil}).\]   
    Let now $F'_{n'}$ be the cdf of $U'_{n'}$.

    By the Berry-Esseen-Theorem, it then holds that 
    \begin{align}
        \sup_{x\in\mathbb{R}} \left|F'_{n'}(x) - \Phi(\frac{x}{\sqrt{\sum_{i=n'}^{\lceil k\cdot n'\rceil} Var(X_i)}})\right| &\leq C_0 \frac{\sum_{i=n'}^{\lceil k\cdot n'\rceil}\mathbb{E}(|X_{i}-\mathbb{E}(X_i)|^3)}{\left(\sum_{i=n'}^{\lceil k\cdot n'\rceil}\mathbb{E}((X_{i} - \mathbb{E}(X_{i}))^2)\right)^{\frac{3}{2}}}
    \end{align}
    For the sake of readability, we define $b(n') = C_0 \frac{\sum_{i=n'}^{\lceil k\cdot n'\rceil}\mathbb{E}(|X_{i}-\mathbb{E}(X_i)|^3)}{\left(\sum_{i=n'}^{\lceil k\cdot n'\rceil}\mathbb{E}((X_{i} - \mathbb{E}(X_{i}))^2)\right)^{\frac{3}{2}}}$.

    For the numerator, the following equality holds:
    \begin{align}
        \sum_{i=n'}^{\lceil k\cdot n'\rceil}\mathbb{E}(|X_{i}-\mathbb{E}(X_i)|^3) &= \sum_{i=n'}^{\lceil k\cdot n'\rceil} \left|q_1[i] - \mathbb{E}(X_i)\right|^3 p + \left|q_2[i]-\mathbb{E}(X_i)\right|^3(1-p)
    \end{align}
    $\mathbb{E}(X_i)$ is a polynomial (see Definition~\ref{definition:x_n}) and different from $q_1[i]$ and $q_2[i]$, as we have non-zero variance. Therefore the numerator is a polynomial and its degree is:
    \begin{align}
        \deg\left(\sum_{i=n'}^{\lceil k\cdot n'\rceil}\mathbb{E}(|X_{i}-\mathbb{E}(X_i)|^3)\right) = 3\deg(\mathcal{P}) + 1
    \end{align}
    Note that the $+1$ comes from the summation and the fact, that $k>1$.

 Analogously, 
    \begin{multline}
        \sum_{i=n'}^{\lceil k\cdot n'\rceil}\mathbb{E}((X_{i} - \mathbb{E}(X_{i}))^2) = \sum_{i=n'}^{\lceil k\cdot n'\rceil} \left(q_1[i] - \mathbb{E}(X_i)\right)^2 p + \left(q_2[i]-\mathbb{E}(X_i)\right)^2(1-p)
    \end{multline}
    and therefore the degree of the denominator is:
    \begin{align}
        \deg\left(\left(\sum_{i=n'}^{\lceil k\cdot n'\rceil}\mathbb{E}((X_{i} - \mathbb{E}(X_{i}))^2)\right)^\frac{3}{2}\right) &=  3 \deg(\mathcal{P}) + \frac{3}{2}
    \end{align}

    Since the degree of the denominator is higher $b(n')$ converges to zero. We denote this

    By definition, $U_{n'} = U'_{n'} +\sum_{i=n'}^{\lceil k\cdot n'\rceil} \mathbb{E}(X_i)$, and therefore $F_{U_n'}(x) = F'_{n'}(x-\sum_{i=n'}^{\lceil k\cdot n'\rceil} \mathbb{E}(X_i))$.

    For the standard normal distribution $\Phi$ it holds that $\Phi(x+\delta)\leq \Phi(x)+|\delta|$, as the probability density function of the standard normal distribution is always smaller than $1$:
    \begin{multline}
        \sup_{x\in\mathbb{R}} \left|F_{U_{n'}}(x) - \Phi\left(\frac{x- \sum_{i=n'}^{\lceil k\cdot n'\rceil} \mathbb{E}(X_i)}{\sqrt{\sum_{i=n'}^{\lceil k\cdot n'\rceil} Var(X_i)}} \right)\right|\\
        \geq \sup_{x\in\mathbb{R}} \left|F_{U_{n'}}(x) - \Phi\left(\frac{x}{\sqrt{\sum_{i=n'}^{\lceil k\cdot n'\rceil} Var(X_i)}} \right)\right| - \frac{\sum_{i=n'}^{\lceil k\cdot n'\rceil} \mathbb{E}(X_i)}{\sqrt{\sum_{i=n'}^{\lceil k\cdot n'\rceil} Var(X_i)}} 
    \end{multline}
    By the above inequality, it then holds that 
    \begin{align}
        \sup_{x\in\mathbb{R}} \left|F_{U_{n'}}(x) - \Phi\left(\frac{x}{\sqrt{\sum_{i=n'}^{\lceil k\cdot n'\rceil} Var(X_i)}} \right)\right| \leq b(n') + \frac{\sum_{i=n'}^{\lceil k\cdot n'\rceil} \mathbb{E}(X_i)}{\sqrt{\sum_{i=n'}^{\lceil k\cdot n'\rceil} Var(X_i)}} 
    \end{align}
    Since $b(n')$ converges to $0$, the whole bound converges when $\deg(Var(X_i)) > 2\deg(\mathbb{E}(X_i)) + 1$.
    
    Analogously, we consider $U'_0 = X_0-\mathbb{E}(X_0)+\dots+X_{n_0}-\mathbb{E}(X_{n_0})$, and let $F'_0$ be its cdf. Following the same argument, it then holds that:

    \begin{align}
        \sup_{x\in\mathbb{R}} \left|F'_{0}(x) - \Phi(\frac{x}{\sqrt{\sum_{i=0}^{n_0} Var(X_i)}})\right| &\leq b(0) := C_0 \frac{\sum_{i=0}^{n_0}\mathbb{E}(|X_{i}-\mathbb{E}(X_i)|^3)}{\left(\sum_{i=0}^{n_0}\mathbb{E}((X_{i} - \mathbb{E}(X_{i}))^2)\right)^{\frac{3}{2}}}
    \end{align} and the expression on the right hand side again converges.

    And since $U_0 = U'_{o} +y_0+\mathbb{E}(X_0)+\dots+\mathbb{E}(X_{n_0})$, and $\Phi(x+\delta)\leq \Phi(x)+|\delta|$, it holds that:
    \begin{multline}
        \sup_{x\in\mathbb{R}} \left|F_{U_{0}}(x) - \Phi\left(\frac{x- \sum_{i=0}^{n_0} \mathbb{E}(X_i)}{\sqrt{\sum_{i=0}^{n_0} Var(X_i)}} \right)\right|\\
        \geq \sup_{x\in\mathbb{R}} \left|F_{U_{0}}(x) - \Phi\left(\frac{x}{\sqrt{\sum_{i=0}^{n_0} Var(X_i)}} \right)\right| - \frac{\sum_{i=0}^{n_0} \mathbb{E}(X_i)}{\sqrt{\sum_{i=0}^{n_0} Var(X_i)}} 
    \end{multline}
    and therefore 
    \begin{align}
        \sup_{x\in\mathbb{R}} \left|F_{U_{0}}(x) - \Phi\left(\frac{x}{\sqrt{\sum_{i=0}^{n_0} Var(X_i)}} \right)\right| \leq b(0) + \frac{\sum_{i=0}^{n_0} \mathbb{E}(X_i)}{\sqrt{\sum_{i=0}^{n_0} Var(X_i)}} 
    \end{align}    

    The bound for the deviation of the cdf of $U_0$ therefore also converges.
    By the definition of convergence, we can always find some $n_0$ for an arbitrary $c_0>0$ (at least numerically), which is bigger than both bounds, as required in the lemma.
\qed
\end{proof}
}

\begin{proof}
    We first consider the zero-mean variables $U_{n'}' = X_{n'} - \mathbb{E}(X_{n'}) + \dots + X_{\lceil n'\cdot k\rceil} - \mathbb{E}(X_{\lceil n'\cdot k\rceil})$. 
    By the Berry-Esseen theorem \cite{berry-esseen-refined-constant}, it  holds: 
    \[\begin{array}{ll}
        \sup_{x\in\mathbb{R}} \left|F'_{U_{n'}}(x) - \Phi(\frac{x}{\sqrt{\sum_{i=n'}^{\lceil k\cdot n'\rceil} Var(X_i)}})\right| &\leq C_0 \frac{\sum_{i=n'}^{\lceil k\cdot n'\rceil}\mathbb{E}(|X_{i}-\mathbb{E}(X_i)|^3)}{\left(\sum_{i=n'}^{\lceil k\cdot n'\rceil}\mathbb{E}((X_{i} - \mathbb{E}(X_{i}))^2)\right)^{\frac{3}{2}}}
    \end{array}\]
    where $C_0$ is a  constant known to be smaller than $0.5591$. Recall that $\Phi$ is the  cdf  of the Normal distribution (Lemma~\ref{definition:z_n}).
    For the sake of readability, we write  $b(n') = C_0 \frac{\sum_{i=n'}^{\lceil k\cdot n'\rceil}\mathbb{E}(|X_{i}-\mathbb{E}(X_i)|^3)}{\left(\sum_{i=n'}^{\lceil k\cdot n'\rceil}\mathbb{E}((X_{i} - \mathbb{E}(X_{i}))^2)\right)^{\frac{3}{2}}}$. 
   As $\deg\left(\sum_{i=n'}^{\lceil k\cdot n'\rceil}\mathbb{E}(|X_{i}-\mathbb{E}(X_i)|^3)\right) = 3\deg(\mathcal{P}) + 1$ and $\deg\left(\left(\sum_{i=n'}^{\lceil k\cdot n'\rceil}\mathbb{E}((X_{i} - \mathbb{E}(X_{i}))^2)\right)^\frac{3}{2}\right) = 3 \deg(\mathcal{P}) + \frac{3}{2}$, the term $b(n')$ converges to $0$. Definition~\ref{def:random:walk:sum} implies that $U_{n'} = U'_{n'} +\sum_{i=n'}^{\lceil k\cdot n'\rceil} \mathbb{E}(X_i)$ and therefore $F_{U_{n'}}(x) = F_{U'_{n'}}(x-\sum_{i=n'}^{\lceil k\cdot n'\rceil} \mathbb{E}(X_i))$.
    By standard properties of the Normal distribution, we have $\Phi(x+\delta)\leq \Phi(x)+|\delta|$. Hence, 
    \[\begin{array}{l}
        \sup_{x\in\mathbb{R}} \left|F_{U_{n'}}(x) - \Phi\left(\frac{x+ \sum_{i=n'}^{\lceil k\cdot n'\rceil} \mathbb{E}(X_i)}{\sqrt{\sum_{i=n'}^{\lceil k\cdot n'\rceil} Var(X_i)}} \right)\right|\\
        \geq ~~~\sup_{x\in\mathbb{R}} \left|F_{U_{n'}}(x) - \Phi\left(\frac{x}{\sqrt{\sum_{i=n'}^{\lceil k\cdot n'\rceil} Var(X_i)}} \right)\right| - \frac{\sum_{i=n'}^{\lceil k\cdot n'\rceil} \mathbb{E}(X_i)}{\sqrt{\sum_{i=n'}^{\lceil k\cdot n'\rceil} Var(X_i)}} 
    \end{array}\]
which implies 
    \begin{align}\label{eq:bound:cdf:conv}
        \sup_{x\in\mathbb{R}} \left|F_{U_{n'}}(x) - \Phi\left(\frac{x}{\sqrt{\sum_{i=n'}^{\lceil k\cdot n'\rceil} Var(X_i)}} \right)\right| \leq b(n') + \frac{\sum_{i=n'}^{\lceil k\cdot n'\rceil} \mathbb{E}(X_i)}{\sqrt{\sum_{i=n'}^{\lceil k\cdot n'\rceil} Var(X_i)}} 
    \end{align}
    The rhs of  inequality~\eqref{eq:bound:cdf:conv} converges to $0$ when $\deg(Var(X_i)) > 2\deg(\mathbb{E}(X_i)) + 1$, and thus~\eqref{eq:cdf:convergence} holds for the zero-mean variables $U_{n'}$. 

    Similarly, we derive a bound for the centered $U'_0$, as well as for $U_0$:
    \begin{equation}\label{eq:cdf_U0_bound}
        \sup_{x\in\mathbb{R}} \left|F_{U_{0}}(x) - \Phi\left(\frac{x}{\sqrt{\sum_{i=0}^{n_0} Var(X_i)}} \right)\right| \leq b_{n_0} + \frac{y_0+\sum_{i=0}^{n_0} \mathbb{E}(X_i)}{\sqrt{\sum_{i=0}^{n_0} Var(X_i)}} 
    \end{equation}
    where $b_{n_0}$ is obtained from the Berry-Esseen theorem as before. 
    
    We set the constant  $c_0$ of~\eqref{eq:cdf:convergence} to be the maximum of the bounds of~\eqref{eq:bound:cdf:conv} and ~\eqref{eq:cdf_U0_bound}.  As such, \eqref{eq:cdf:convergence} holds. 
    \qed
\end{proof}


\lemmaSummationTailBound*
\LongProof{
\begin{proof}
    We again consider the centered version of $U'_{n'}$ as in the proof of Lemma~\ref{lemma:almost-normal}. Let $q_1'[i] = q_1[i]-E[X_i]$ and $q_2'[i] = q_2[i]-E[X_i]$ be the increments of $U'_{n'}$.
    
    The moment generating function of $U'_{n'}$ can be bounded using Hoeffding's Lemma~\cite[Lemma~2.6]{massart2007concentration}, since each variable $X_i-\mathbb{E}(X_i)$ has bounded support:

    \begin{multline}
        M_{U'_{n'}}(t) = \prod_{i=n'}^{\lceil n'\cdot k\rceil}\exp\left(tq_1'[i]\right)p + \exp\left(tq_2'[i]\right)(1-p)\\\leq\prod_{i=n'}^{\lceil n'\cdot k\rceil} \exp\left(\frac{t^2(q_1'[i]-q_2'[i])^2}{8}\right)
    \end{multline}
    
    Since the increments are zero-mean, $q_2'[i] = -\frac{q_1'[i]p}{(1-p)}$. Therefore,
    \begin{align}
    \frac{(q_1'[i] - q_2'[i])^2}{Var(X_i)} = \frac{\left(q_1'[i]+\frac{q_1'[i]p}{(1-p)}\right)^2}{q_1'[i]^2p+\left(\frac{q_1'[i]p}{(1-p)}\right)^2(1-p)} &=  \frac{\left(1+\frac{p}{(1-p)}\right)^2}{p+\left(\frac{p}{(1-p)}\right)^2(1-p)}\\
    &=\frac{1}{(1-p)p}
    \end{align}
    Hence,
    \begin{align}
    M_{U'_{n'}}(t) &\leq\prod_{i=n'}^{\lceil n'\cdot k\rceil} \exp\left(\frac{t^2Var(X_i)}{8((1-p)p)}\right)\\&=\exp\left(\frac{t^2\sum_{i=n'}^{\lceil n'\cdot k\rceil}Var(X_i)}{8((1-p)p)}\right)=\exp\left(\frac{t^2Var(U_{n'}')}{8((1-p)p)}\right)
    \end{align}

    The Chernoff bound\cite[p.77]{lattimore2020bandit} states, that for any positive $t$ and variable $X$, with its moment generating function $M_{X}(t)$ it holds that:
    \begin{equation}
        \mathbb{P}(X\geq a) \leq e^{-ta}M_{X}(t).
    \end{equation}

    We now apply the Chernoff-bound for the variable $U'_{n'}$ and choose $a=\lambda \sqrt{Var(U'_{n'})}$ and $t=\frac{\lambda4((1-p)p)}{\sqrt{Var(U'_{n'})}}$:
    \begin{align}
        \mathbb{P}(X\geq \lambda\sqrt{Var(U'_{n'})}) \leq \exp\left(4((1-p)p) \left(-\frac{\lambda^2}{2}\right)\right)
    \end{align}
    By the definition of $U'_{n'}$ (note tha also $Var(U'_{n'}) = Var(U_{n'})$), it then follows that 
    \begin{multline}
        \mathbb{P}(U_{0}\geq \lambda\sqrt{Var(U_{n'})}+\sum_{i=n'}^{\lceil n'\cdot k\rceil}\mathbb{E}(X_i)) \leq \exp\left(4((1-p)p) \left(-\frac{\lambda^2}{2}\right)\right)
    \end{multline}

    Analogously, for $U_0'$ we obtain can obtain the same bound: \begin{align}
        M_{U'_0}(t)\leq \exp\left(\frac{t^2Var(U_0'}{8(1-p)p}\right).
    \end{align}
    Then applying the Chernoff-bound:
    \begin{align}
        \mathbb{P}(X\geq \lambda\sqrt{Var(U'_{0})}) \leq \exp\left(4((1-p)p) \left(-\frac{\lambda^2}{2}\right)\right)
    \end{align}
    and since $U_0=U_0'+y_0+\sum_{i=0}^{n_0} \mathbb{E}(X_i)$:
    \begin{multline}
        \mathbb{P}(U_{0}\geq \lambda\sqrt{Var(U_{0})}+y_0+\sum_{i=0}^{n_0}\mathbb{E}(X_i)) \leq \exp\left(4((1-p)p) \left(-\frac{\lambda^2}{2}\right)\right)
    \end{multline}
     By the same argument as in the proof of Lemma~\ref{lemma:almost-normal} \begin{align}\delta_1 =\max\left\{ \frac{\sum_{i=n'}^{\lceil n'\cdot k\rceil}\mathbb{E}(X_i)}{\sqrt{\sum_{i=n'}^{\lceil n'\cdot k\rceil}Var(X_i)}}, \frac{y_0+\sum_{i=0}^{n_0}\mathbb{E}(X_i)}{\sqrt{\sum_{i=0}^{n_0}Var(X_i)}}\right\}\end{align} converges to zero.
\qed
\end{proof}
}

\begin{proof}
    Similarly to Lemma~\ref{lemma:almost-normal}, we consider the centered version of $U'_{n'}$. Let $q_1'[i] = q_1[i]-E[X_i]$ and $q_2'[i] = q_2[i]-E[X_i]$ be the increments of $U'_{n'}$. The moment generating function (mgf) of $U'_{n'}$, denoted as $M_{U'_{n'}}$, is bounded by the Hoeffding lemma~\cite[Lemma~2.6]{massart2007concentration}, since each variable $X_i-\mathbb{E}(X_i)$ has bounded support. As the (loop) increments are zero-mean, we have:  
    \begin{equation*}
        M_{U'_{n'}}(t) = \prod_{i=n'}^{\lceil n' k\rceil}M_{X_i}(t)\leq\prod_{i=n'}^{\lceil n' k\rceil} \exp\left(\frac{t^2(q_1'[i]-q_2'[i])^2}{8}\right)=\exp\left(\frac{t^2Var(U_{n'})}{8(1-p)p}\right)
    \end{equation*}
    Recall that the Chernoff bound\cite[p.77]{lattimore2020bandit} states that  $\mathbb{P}(X\geq a) \leq e^{-ta}M_{X}(t)$ holds for a random variable $X$ and its  mfg $M_{X}(t)$, for  all $t$. We can then bound $U_{n'}$ with $t=\frac{\lambda4((1-p)p)}{\sqrt{Var(U'_{n'})}}$ and get:
    \begin{equation*}
        \mathbb{P}(X\geq \lambda\sqrt{Var(U'_{n'})}) \leq \exp\left(4(1-p)p \left(-\frac{\lambda^2}{2}\right)\right)
    \end{equation*}
    Using Definition~\ref{def:random:walk:sum}, we  derive: 
    \[\mathbb{P}(U_{n'}\geq \lambda\sqrt{Var(U'_{n'})}+\sum_{i=n'}^{\lceil n'\cdot k\rceil}\mathbb{E}(X_i)) \leq \exp\left(4((1-p)p) \left(-\frac{\lambda^2}{2}\right)\right).\]
    
    \noindent Similarly,  we  bound $U_0$. Then,  we set $\delta_1 =\max\left\{ \frac{\sum_{i=n'}^{\lceil n'\cdot k\rceil}\mathbb{E}(X_i)}{\sqrt{\sum_{i=n'}^{\lceil n'\cdot k\rceil}Var(X_i)}},\frac{y_0+\sum_{i=0}^{n_0}\mathbb{E}(X_i)}{\sqrt{\sum_{i=0}^{n_0}Var(X_i)}}\right\}$, which converges to $0$ with $n_0$ increasing.
    \qed
\end{proof}

\label{appendix:proof-variance-growth}
\lemmaVarianceGrowth*
\begin{proof}
    Let $q_\text{var}[n]$ be the polynomial describing the variance of $X_1+\dots + X_n$. Let $m=deg(q_\text{var})$ denote its degree. Then $m=\deg(\mathcal{P})+1$ and $Var(Z_n) =  q_\text{var}[\lceil k\cdot n'_{(n)} \rceil] - q_\text{var}[n'_{(n)}]$. Since $Z_1 + \dots + Z_{n-1} = X_1 + \dots +X_{n'_{(n)}}$ and thus $Var(Z_1 + \dots + Z_{n-1}) = q_\text{var}[n'_{(n)}]$.

    But then the factor $d$, such that $Var(Z_n) \geq d (Var(Z_1) +\dots+ Var(Z_{n-1}))=d(Var(Z_1 +\dots+ Z_{n-1}))$ can be computed:
    \begin{align*}
        q_\text{var}[\lceil k\cdot n'_{(n)} \rceil] - q_\text{var}[n'_{(n)}] \geq d \cdot  q_\text{var}[n'_{(n)}]
    \end{align*}

    A polynomial can be bounded by its leading term with a multiplicative factor. Specifically, with $\delta_1\delta_2\in\mathbb{R^+}$:
    \[\forall n\geq n_0'(\delta_1,\delta_2) : (1-\delta_1)a_m n^m \leq a_1x + \dots + a_mn^m \leq (1+\delta_2) a_mn^m\]

   Inserting this in the above equation:
    \begin{align*}
        (1-\delta_1)a_m (\lceil k\cdot n'_{(n)} \rceil)^m &\geq (d+1)  (1+\delta_2) a_m n'^m_{(n)}\\
        \delta'k^m=\frac{(1-\delta_1)}{1+\delta_2}k^m &\geq (d+1)
        \tag*{\qed}
    \end{align*}
\end{proof}

\lemmaBoundFromTerminationProb*
\begin{proof}
      During the $n$th iteration of the inner loop body,  Figure~\ref{alg:random_walk_transformed}  could have terminated $\lfloor\log_{k+\tau}(n)-\lceil\log_{k+\tau}( n_0)\rceil\rfloor \leq \log_{k+\tau}(n) - \log_{k+\epsilon}(n_0) - 2$ times. The error $\tau$ here accounts for the number of rounding ups, as we take $\lceil kn'_{(n)}\rceil$. A safe choice to approximate $\tau$  is $\tau = \frac{1}{n_0}$. With increasing $n_0$, note that $\tau$ converges to $0$.
As such, the probability of 
Figure~\ref{alg:random_walk_transformed} not  terminating is bounded:
    \begin{align*}
        \mathbb{P}(T\geq n) &\leq (1-p_\text{term})^{\log_{k+\tau}(n) - \log_{k+\tau}(n_0) - 2}
        = \frac{n^{\frac{\ln(1-p_\text{term})}{\ln(k+\tau)}}}{(1-p_\text{term})^{\log_{k+\tau}(n_0)+2}} 
        \tag*{\qed}
    \end{align*}
\end{proof}

\theoremPolynomialRandomWalkThreshold*
\begin{proof}
   As
    $\deg(\mathcal{P})>\deg(p(q_1[n])+(1-p)q_2[n])$, we get  $\deg(Var(X_i)) > 2\deg(\mathbb{E}(X_i)) + 1$. We use  Lemmas~\ref{lemma:almost-normal}--\ref{lemma:tail-bound} and 
    take an  arbitrary $d$. By Theorem~\ref{theorem:epsilon-exists-always}, there exists an inductive bound with a nonzero termination probability $p_\text{term}$.   Lemma~\ref{lemma:bound-from-termination-probability} implies that the expected stopping time  $\mathbb{E}(T)$ is finite  when $\ln(1-p_\text{term})<-\ln(k+\frac{1}{n_0})$. Further, Lemma~\ref{lemma:variance-growth-polynomial-random-walk} asserts  $d=\delta'k^{\deg(\mathcal{P})+1}-1$. Since $\ln(k+\frac{1}{n_0})$ converges to $\ln(k)$ as $n_0$ increases, we can conclude that 
 $\mathcal{P}$ has finite expected stopping time. 
    \qed
\end{proof}

\section{Implementation Details}\label{appendix:details-inductive-bound-model}
We provide further details on the linear model used in our  implementation, as discussed in Section~\ref{section:implementation}. 

\paragraph{\bf Inductive bound sets.}
To compute an inductive bound $B$, the parameters $\epsilon, d,c_0, C_1, \delta_1$ must be fixed. Additionally, we require a sorted vector $\vec{a}$ with length $m$, specifying $\vec{a}_B$, as well as a sorted vector $\vec{c}$ with length $k$, with values, for which the Chernoff bound should be used. The variables $\vec{b}$, corresponding to $\vec{b}_B$ are the variables of interest. 

We introduce the auxiliary variables $\vec{d}$ of length $m+k$, which instead represent a bound for probability mass of the interval $[\vec{a}_{i-1};\vec{a}_{i}]$ and $[\vec{c}_{i-1};\vec{c}]$ of the conditioned variable $(S_{n}|S_n\geq F_{S_n}^{-1}(\epsilon))$:
\begin{align*}
    \vec{d}_1 (1-\epsilon) \leq \vec{b}_1-\epsilon\land
    \forall_{2\leq i\leq m}: \vec{d}_i(1-\epsilon)\leq \vec{b}_{i}-\vec{b}_{i-1}
\end{align*}
and similarly for the values of the Chernoff bound with $b =\frac{ \sqrt{2 \ln{\frac{1}{1-\epsilon}}}}{C_1(\sqrt{(1+d)}-1)}$  (also $\vec{d}_{m+1}$ represents the interval $[\vec{a}_{m};\vec{c}_0]$):
\begin{multline*}
\vec{d}_{m+1}(1-\epsilon) \leq 1- \exp\left(C_1\frac{-(\vec{c}_{i} - b-\delta_1)^2}{2}\right)-\vec{b}_m \land
    \forall_{2\leq i\leq k}: \\\vec{d}_{m+i}(1-\epsilon)\leq \left(\exp\left(C_1\frac{-(\vec{c}_{i-1} - b-\delta_1)^2}{2}\right)-\exp\left(C_1\frac{-(\vec{c}_{i} - b-\delta_1)^2}{2}\right)\right)
\end{multline*}

Every inequality must then hold for $S_{n+1}$:
\begin{multline*}
    \vec{b}_{i}\geq \sum_{j=1}^{m} \vec{d}_j \left(\Phi\left(\frac{\vec{a}_i\sqrt{1+d}-\vec{a}_j}{\sqrt{d}}\right)-c_0\right) \\+ \sum_{j=1}^{k} \vec{d}_{m+j}\left(\Phi\left(\frac{\vec{a}_i\sqrt{1+d}-\vec{c}_j}{\sqrt{d}}\right)-c_0\right)
\end{multline*}
for all $i=1, \ldots, m$.

And ultimately, we can only be less restrictive than the loop exit condition, hence the whole tail that is removed when conditioning must be smaller $0$:
\begin{align*}
    \vec{a_1}\leq 0 \land \vec{b}_1\geq \epsilon 
\end{align*}

To ensure, that $S_0$ satisfies the bound, it must also hold, that: 
\begin{align*}
    (\Phi(\vec{a}_i) - c_0)\geq \vec{b}_i
\end{align*}
for all $i=1,\ldots, m$.


\section{Numeric and Experimental Results}\label{appendix:example-inductive-bound-for-linear}
Here we give details on inductive bound used for proving PAST of linear random walks in our experiments from Section~\ref{section:implementation}, in particular on Example~\ref{example:past-of-example-program}. Furthermore, in Appendix~\ref{appendix:converging-terms-influence-actual-stopping-time}, we argue, that the properties of the polynomials which influence $n_0$ in our proofs have a similar effect on the empirical stopping time. The consistency of the results, along with the effect of the granularity of the inductive bound is shown in Appendix~\ref{appendix:performance-genetic-algorithm}

\subsection{Inductive bound used for proving PAST of linear random walk}
For the parameters $\epsilon=0.11286862346080692,d=0.410143812649425, c_0=10^{-8},C_1=1, \delta_1=10^{-5} $ the bound given in $\vec{a}$ and $\vec{b}$ given in Table~\ref{table:exact-inductive-bound-1}~and~\ref{table:exact-inductive-bound-2}, using $\vec{c} = (6.497321214442595)$, is inductive. The different components of the bound are shown in Figure~\ref{fig:plot_bound_S_n}.
\begin{figure}
\centering
\includegraphics[width=0.8\textwidth]{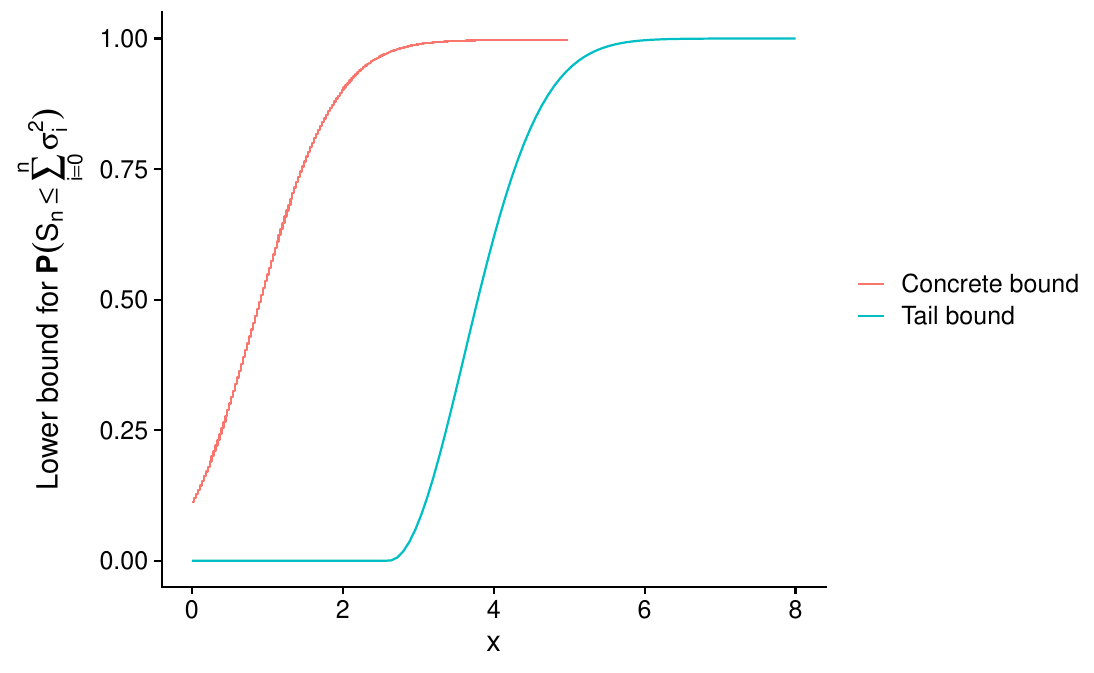}
\caption{Bound for proving PAST of linear random walk} \label{fig:plot_bound_S_n}
\end{figure}

\begin{table}
    \centering
    \caption{Inductive bound for proving PAST of linear random walk (1)}
\begin{minipage}[t]{0.32\textwidth}
\begin{tabular}{|l|l|}
\hline
$\vec{a}$ & $\vec{b}$ \\
\hline
0.00000000 & 0.11296053\\
0.02610488 & 0.12030958\\
0.05220976 & 0.12796144\\
0.07831464 & 0.13591606\\
0.10441952 & 0.14417252\\
0.13052440 & 0.15272904\\
0.15662928 & 0.16158298\\
0.18273416 & 0.17073082\\
0.20883904 & 0.18016817\\
0.23494392 & 0.18988980\\
0.26104880 & 0.19988960\\
0.28715368 & 0.21016063\\
0.31325856 & 0.22069516\\
0.33936344 & 0.23148466\\
0.36546832 & 0.24251983\\
0.39157320 & 0.25379067\\
0.41767808 & 0.26528648\\
0.44378296 & 0.27699594\\
0.46988784 & 0.28890709\\
0.49599272 & 0.30100746\\
0.52209760 & 0.31328405\\
0.54820248 & 0.32572342\\
0.57430736 & 0.33831174\\
0.60041224 & 0.35103482\\
0.62651712 & 0.36387817\\
0.65262200 & 0.37682710\\
0.67872688 & 0.38986672\\
0.70483176 & 0.40298200\\
0.73093664 & 0.41615786\\
0.75704152 & 0.42937920\\
0.78314640 & 0.44263094\\
0.80925128 & 0.45589810\\
0.83535616 & 0.46916584\\
0.86146104 & 0.48241946\\
0.88756592 & 0.49564453\\
0.91367080 & 0.50882687\\
0.93977568 & 0.52195260\\
0.96588056 & 0.53500818\\
0.99198544 & 0.54798048\\
1.01809032 & 0.56085673\\
\hline
\end{tabular}
\end{minipage}
\hfill
\begin{minipage}[t]{0.32\textwidth}
\begin{tabular}{|l|l|}
\hline
$\vec{a}$ & $\vec{b}$ \\
\hline
1.04419520 & 0.57362465\\
1.07030008 & 0.58627239\\
1.09640495 & 0.59878861\\
1.12250983 & 0.61116247\\
1.14861471 & 0.62338366\\
1.17471959 & 0.63544241\\
1.20082447 & 0.64732951\\
1.22692935 & 0.65903632\\
1.25303423 & 0.67055477\\
1.27913911 & 0.68187737\\
1.30524399 & 0.69299720\\
1.33134887 & 0.70390794\\
1.35745375 & 0.71460384\\
1.38355863 & 0.72507972\\
1.40966351 & 0.73533098\\
1.43576839 & 0.74535360\\
1.46187327 & 0.75514407\\
1.48797815 & 0.76469948\\
1.51408303 & 0.77401740\\
1.54018791 & 0.78309596\\
1.56629279 & 0.79193376\\
1.59239767 & 0.80052993\\
1.61850255 & 0.80888403\\
1.64460743 & 0.81699609\\
1.67071231 & 0.82486659\\
1.69681719 & 0.83249642\\
1.72292207 & 0.83988686\\
1.74902695 & 0.84703957\\
1.77513183 & 0.85395658\\
1.80123671 & 0.86064024\\
1.82734159 & 0.86709325\\
1.85344647 & 0.87331856\\
1.87955135 & 0.87931943\\
1.90565623 & 0.88509937\\
1.93176111 & 0.89066212\\
1.95786599 & 0.89601164\\
1.98397087 & 0.90115207\\
2.01007575 & 0.90608773\\
2.03618063 & 0.91082312\\
2.06228551 & 0.91536286\\
\hline
\end{tabular}
\end{minipage}
\hfill
\begin{minipage}[t]{0.32\textwidth}
\begin{tabular}{|l|l|}
\hline
$\vec{a}$ & $\vec{b}$ \\
\hline
2.08839039 & 0.91971167\\
2.11449527 & 0.92387440\\
2.14060015 & 0.92785598\\
2.16670503 & 0.93166140\\
2.19280991 & 0.93529570\\
2.21891479 & 0.93876397\\
2.24501967 & 0.94207130\\
2.27112455 & 0.94522282\\
2.29722943 & 0.94822363\\
2.32333431 & 0.95107883\\
2.34943919 & 0.95379348\\
2.37554407 & 0.95637261\\
2.40164895 & 0.95882119\\
2.42775383 & 0.96114415\\
2.45385871 & 0.96334635\\
2.47996359 & 0.96543255\\
2.50606847 & 0.96740745\\
2.53217335 & 0.96927568\\
2.55827823 & 0.97104173\\
2.58438311 & 0.97271002\\
2.61048799 & 0.97428488\\
2.63659287 & 0.97577049\\
2.66269775 & 0.97717096\\
2.68880263 & 0.97849027\\
2.71490751 & 0.97973227\\
2.74101239 & 0.98090071\\
2.76711727 & 0.98199922\\
2.79322215 & 0.98303130\\
2.81932703 & 0.98400034\\
2.84543191 & 0.98490961\\
2.87153679 & 0.98576224\\
2.89764167 & 0.98656127\\
2.92374655 & 0.98730960\\
2.94985143 & 0.98801003\\
2.97595631 & 0.98866522\\
3.00206119 & 0.98927774\\
3.02816607 & 0.98985004\\
3.05427095 & 0.99038445\\
3.08037583 & 0.99088322\\
3.10648071 & 0.99134848\\
\hline
\end{tabular}
\end{minipage}

    \label{table:exact-inductive-bound-1}
\end{table}

\begin{table}
    \centering
    \caption{Inductive bound for proving PAST of linear random walk (2)}
\begin{minipage}[t]{0.32\textwidth}
\vspace{0pt}
\hfill
\begin{tabular}{|l|l|}
\hline
$\vec{a}$ & $\vec{b}$ \\
\hline
3.13258559 & 0.99178224\\
3.15869047 & 0.99218643\\
3.18479535 & 0.99256290\\
3.21090023 & 0.99291337\\
3.23700511 & 0.99323951\\
3.26310998 & 0.99354286\\
3.28921486 & 0.99382491\\
3.31531974 & 0.99408706\\
3.34142462 & 0.99433062\\
3.36752950 & 0.99455684\\
3.39363438 & 0.99476689\\
3.41973926 & 0.99496188\\
3.44584414 & 0.99514284\\
3.47194902 & 0.99531075\\
3.49805390 & 0.99546653\\
3.52415878 & 0.99561103\\
3.55026366 & 0.99574506\\
3.57636854 & 0.99586937\\
3.60247342 & 0.99598467\\
3.62857830 & 0.99609162\\
3.65468318 & 0.99619082\\
3.68078806 & 0.99628285\\
3.70689294 & 0.99636825\\
3.73299782 & 0.99644751\\
3.75910270 & 0.99652109\\
3.78520758 & 0.99658943\\
3.81131246 & 0.99665292\\
3.83741734 & 0.99671194\\
3.86352222 & 0.99676682\\
3.88962710 & 0.99681789\\
3.91573198 & 0.99686545\\
3.94183686 & 0.99690976\\
3.96794174 & 0.99695110\\
3.99404662 & 0.99698968\\
4.02015150 & 0.99702574\\
4.04625638 & 0.99705948\\
4.07236126 & 0.99709109\\
4.09846614 & 0.99712074\\
4.12457102 & 0.99714861\\
4.15067590 & 0.99717484\\
\hline
\end{tabular}
\end{minipage}
\hspace{0.03\textwidth}
\begin{minipage}[t]{0.32\textwidth}
\vspace{0pt}
\begin{tabular}{|l|l|}
\hline
$\vec{a}$ & $\vec{b}$ \\
\hline
4.17678078 & 0.99719959\\
4.20288566 & 0.99722298\\
4.22899054 & 0.99724515\\
4.25509542 & 0.99726622\\
4.28120030 & 0.99728630\\
4.30730518 & 0.99730550\\
4.33341006 & 0.99732393\\
4.35951494 & 0.99734168\\
4.38561982 & 0.99735885\\
4.41172470 & 0.99737553\\
4.43782958 & 0.99739181\\
4.46393446 & 0.99740777\\
4.49003934 & 0.99742349\\
4.51614422 & 0.99743905\\
4.54224910 & 0.99745452\\
4.56835398 & 0.99746999\\
4.59445886 & 0.99748552\\
4.62056374 & 0.99750118\\
4.64666862 & 0.99751703\\
4.67277350 & 0.99753314\\
4.69887838 & 0.99754957\\
4.72498326 & 0.99756637\\
4.75108814 & 0.99758359\\
4.77719302 & 0.99760130\\
4.80329790 & 0.99761953\\
4.82940278 & 0.99763833\\
4.85550766 & 0.99765773\\
4.88161254 & 0.99767776\\
4.90771742 & 0.99769847\\
4.93382230 & 0.99771986\\
4.95992718 & 0.99774197\\
4.98603206 & 0.99776480\\
\hline
\end{tabular}
\end{minipage}
\hfill

    \label{table:exact-inductive-bound-2}
\end{table}

\subsection{Relation of converging terms and empirical stopping time}
\label{appendix:converging-terms-influence-actual-stopping-time}
While the relation of the degree of the polynomials and the slope of the empirically measured $\mathbb{P}(T\geq n)$ has been studied in Section~\ref{subsection:experiments}, the effects of other properties of the polynomials, which cause the terms in the proofs to converge slower, can also be seen in the same kind of plot.

\begin{figure}
\subfloat[with different initial value\label{fig:empirical-initial-values}]{%
  \includegraphics[width=0.5\textwidth]{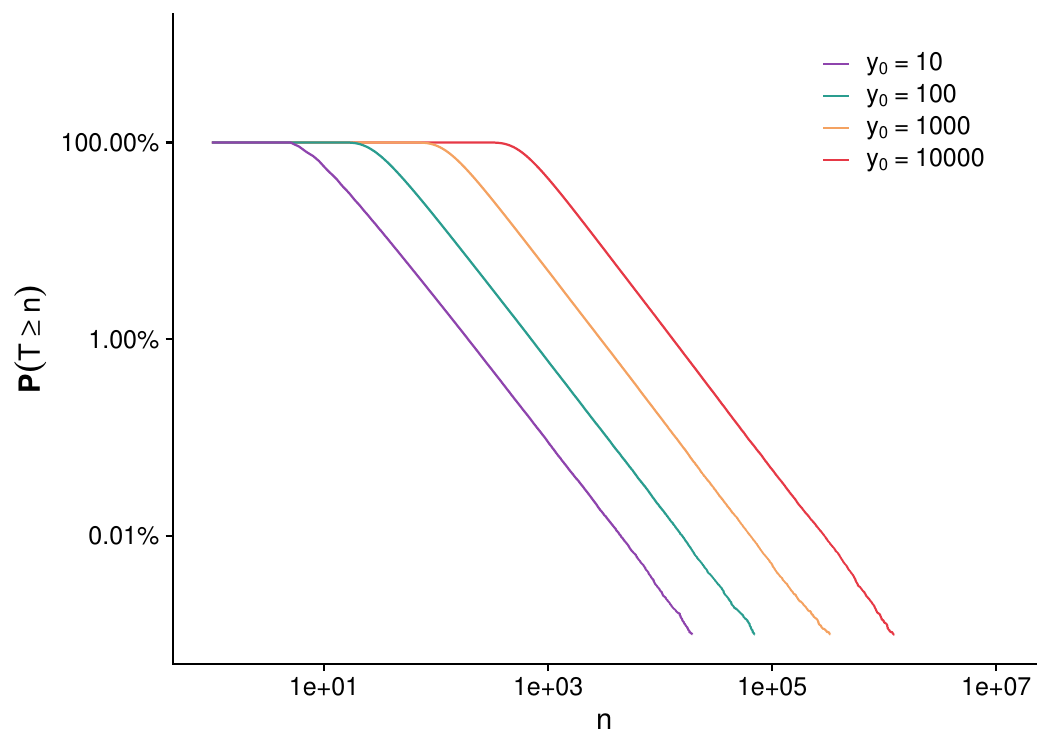}%
}\hfil
\subfloat[with different drift\label{fig:empirical-drift}]{%
  \includegraphics[width=0.5\textwidth]{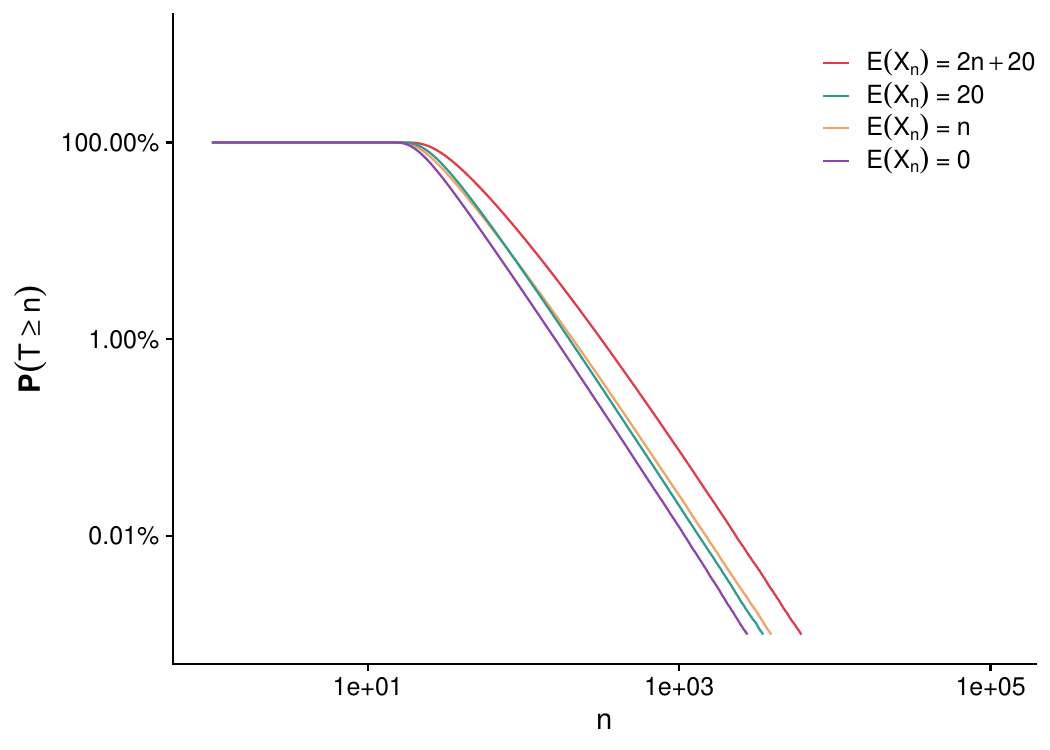}%
}\hfil
\subfloat[with different percentages\label{fig:empirical-different-percentages-appendix}]{%
  \includegraphics[width=0.5\textwidth]{figures/different_percentages.pdf}%
}\hfil
\subfloat[with change of slope\label{fig:slope-change}]{%
\includegraphics[width=0.5\textwidth]{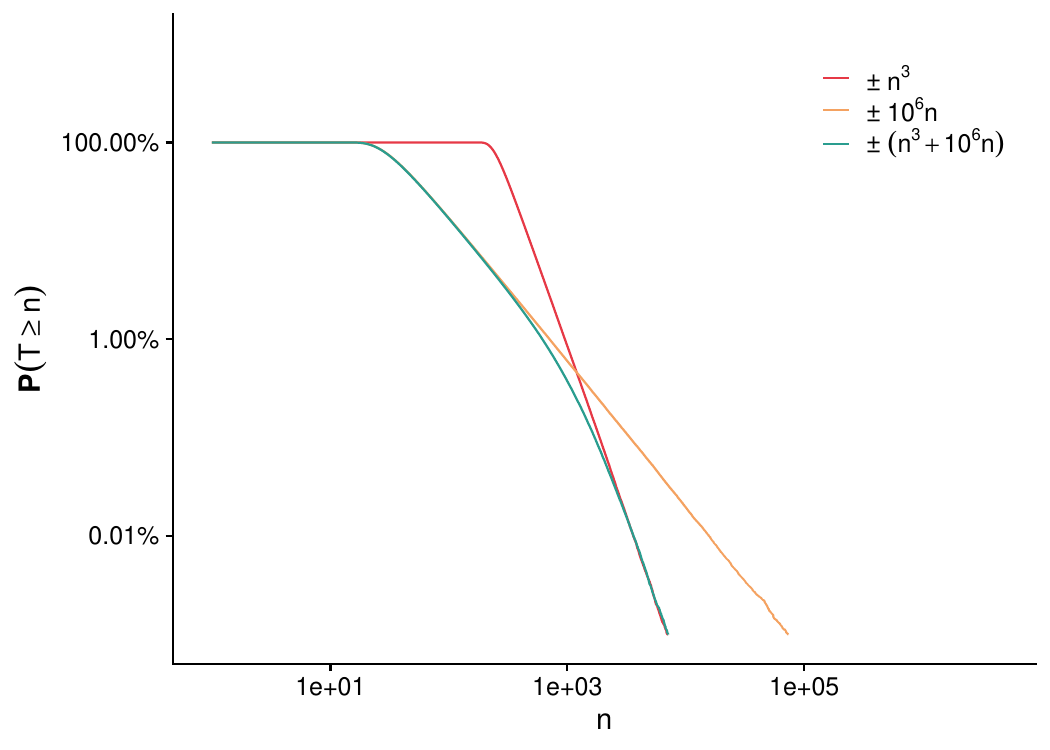}}%

\caption{Empirical results on stopping times of polynomial random walks}
\label{fig:more-empirical-survival-rates}
\end{figure}

Different initial values are shown in Figure~\ref{fig:empirical-initial-values}, and show that termination happens later, but $\mathbb{P}(T\geq n)$ decays with the same exponent. In our implementation this is accounted for through $c_0$ (for $U_0$) and $\delta_1$ in Lemma~\ref{lemma:tail-bound}.
Figure~\ref{fig:empirical-drift} shows random walks with $\pm n^2$ as leading term, and a (strictly positive) polynomial, which is the ``drift'' of the random walk. Even the random walk with the high linear drift seems to eventually converge to a line with the same slope as the other random walks. This effect is covered through the same constants.

The influence of different values of $p$ is, in addition to $C$ in the tail-bound, covered through $c_0$ in Lemma~\ref{lemma:almost-normal}. Our method computes the bound based on the value of $p(1-p)$ (e.g. the bound for $p=0.1$ and $p=0.9$ is the same), since the absolute deviation is captured, and small values of $p$ cause a higher deviation, albeit in the negative,``favorable'' direction. This unsharpness can be clearly seen in Figure~\ref{fig:empirical-different-percentages-appendix}.

A high non-leading term influences the speed of convergence of $\delta'$ in Lemma~\ref{lemma:variance-growth-polynomial-random-walk}. The empirical stopping time then first resembles that of the non-leading term, while later converging towards that of the leading term, as displayed in Figure~\ref{fig:slope-change}. This also supports the claim, that the explicit stopping time could better be approximated, if multiple segments were used.

\subsection{Performance of genetic algorithm}
\label{appendix:performance-genetic-algorithm}
For the genetic algorithm, higher values of $g$ increase the size of the linear model and hence the computation time. To mitigate this effect, $g$ is chosen relatively small in the beginning, but its value increases with each generation. Similarly, the population size shrinks. Usually, higher values of $g$ allow to decrease $d$ and increase $\epsilon$ slightly. Intuitively, the algorithm tries to first find ``good values'' for the parameters, which in later generations are then improved further. Figure~\ref{fig:performance-different-genetic-configurations} shows the relative performance and the running time for different specifications after $50$ generations, where the population size and the granularity $g$ are either set to a fixed value, or set to linearly decrease and increase respectively. Figure~\ref{fig:genetic-quality} shows that (i) the results of our method are relatively reliably and (ii) changing the parameters over time causes the the absolute value of the exponent of the bound to be slightly lower, but reduces the running time drastically, as displayed in Figure~\ref{fig:genetic-running-time}. This gives reason to making a population, linearly decreasing from $100$ to $20$, and a granularity linearly increasing from $50$ to $200$ the default values in our implementation.

\begin{figure}
\subfloat[solution quality\label{fig:genetic-quality}]{%
  \includegraphics[width=0.45\textwidth]{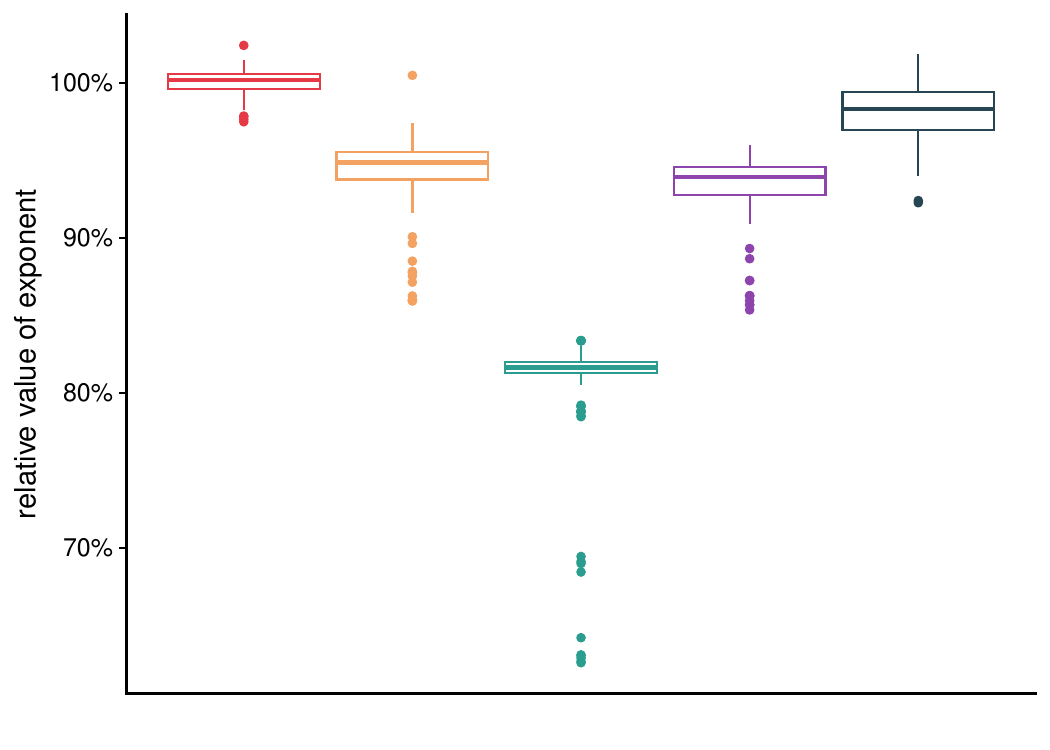}%
}\hfil
\subfloat[computation time\label{fig:genetic-running-time}]{%
  \includegraphics[width=0.45\textwidth]{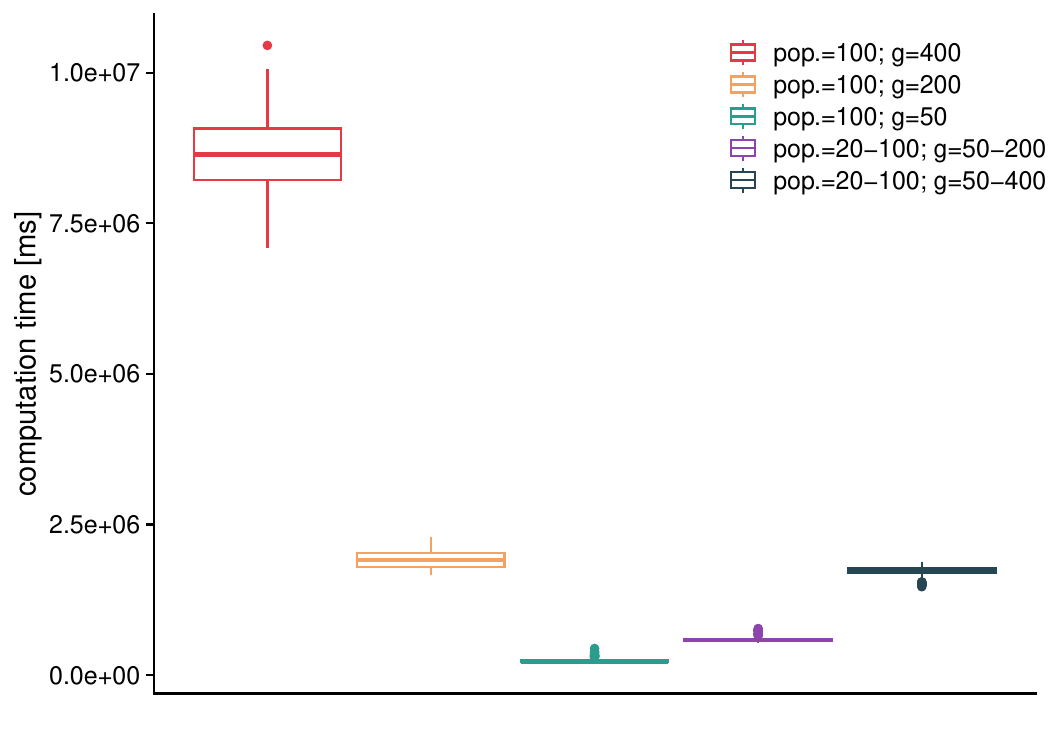}%
}\hfil

\caption{Performance of different genetic algorithm configurations}
\label{fig:performance-different-genetic-configurations}
\end{figure}

\end{document}